\documentclass[11pt]{article}

\usepackage[margin=1in]{geometry}

\usepackage{amssymb,amsmath,amsthm,amsfonts}
\usepackage{mathtools}
\usepackage{enumitem}
\usepackage[numbers,comma,sort&compress]{natbib}
\usepackage{authblk}
\usepackage{graphicx}
\usepackage[font=small]{caption}
\usepackage[labelformat=simple]{subcaption}
\usepackage{float}
\usepackage[ruled,vlined,linesnumbered]{algorithm2e}
\usepackage{physics}
\usepackage{footnote}
\usepackage{xcolor}

\usepackage[colorlinks]{hyperref}

\newtheorem{theorem}{Theorem}

\newtheorem{lemma}{Lemma}
\newtheorem{proposition}{Proposition}

\newtheorem{corollary}{Corollary}
\newtheorem{remark}{Remark}

\newcommand{\eqn}[1]{(\ref{eqn:#1})}

\newcommand{\rem}[1]{\hyperref[rem:#1]{Remark~\ref*{rem:#1}}}
\newcommand{\thm}[1]{\hyperref[thm:#1]{Theorem~\ref*{thm:#1}}}
\newcommand{\cor}[1]{\hyperref[cor:#1]{Corollary~\ref*{cor:#1}}}
\newcommand{\defn}[1]{\hyperref[defn:#1]{Definition~\ref*{defn:#1}}}
\newcommand{\lem}[1]{\hyperref[lem:#1]{Lemma~\ref*{lem:#1}}}
\newcommand{\prop}[1]{\hyperref[prop:#1]{Proposition~\ref*{prop:#1}}}
\newcommand{\fig}[1]{\hyperref[fig:#1]{Figure~\ref*{fig:#1}}}
\newcommand{\tab}[1]{\hyperref[tab:#1]{Table~\ref*{tab:#1}}}
\newcommand{\algo}[1]{\hyperref[algo:#1]{Algorithm~\ref*{algo:#1}}}
\renewcommand{\sec}[1]{\hyperref[sec:#1]{Section~\ref*{sec:#1}}}
\newcommand{\append}[1]{\hyperref[append:#1]{Appendix~\ref*{append:#1}}}
\newcommand{\fac}[1]{\hyperref[fac:#1]{Fact~\ref*{fac:#1}}}
\newcommand{\lin}[1]{\hyperref[lin:#1]{Line~\ref*{lin:#1}}}
\newcommand{\fnote}[1]{\hyperref[fnote:#1]{Footnote~\ref*{fnote:#1}}}

\newcommand\vecc[1]{\mathbf{#1}}

\SetKwInput{KwInput}{Input}
\SetKwInput{KwOutput}{Output}

\def\>{\rangle}
\def\<{\langle}

\def\trans{^{\top}}

\newcommand{\B}{\mathbb{B}}

\newcommand{\R}{\mathbb{R}}
\newcommand{\C}{\mathbb{C}}

\newcommand{\E}{\mathbb{E}}

\newcommand{\A}{X}
\newcommand{\x}{w}

\let\var\relax

\DeclareMathOperator{\clip}{clip}
\DeclareMathOperator{\poly}{poly}
\DeclareMathOperator{\var}{Var}

\DeclareMathOperator{\MEB}{MEB}
\DeclareMathOperator{\SVM}{SVM}

\newcommand{\range}[1]{[#1]}

\newcommand{\hd}[1]{\vspace{4mm} \noindent \textbf{#1}}
\def \eps {\epsilon}

\begin{document}

\title{Sublinear quantum algorithms for training linear and\\kernel-based classifiers}

\author{Tongyang Li\thanks{tongyang@cs.umd.edu}\qquad\quad Shouvanik Chakrabarti\thanks{shouv@cs.umd.edu}\qquad\quad Xiaodi Wu\thanks{xwu@cs.umd.edu} \\
\small{Department of Computer Science, Institute for Advanced Computer Studies, and \\ Joint Center for Quantum Information and Computer Science, University of Maryland}}

\normalsize
\renewcommand\Authsep{  }
\renewcommand\Authands{ }

\date{}

\maketitle
\thispagestyle{empty}

\begin{abstract}
We investigate quantum algorithms for classification, a fundamental problem in machine learning, with provable guarantees. Given $n$ $d$-dimensional data points, the state-of-the-art (and optimal) classical algorithm for training classifiers with constant margin \cite{clarkson2012sublinear} runs in $\tilde{O}(n+d)$\footnote{$\tilde{O}(\cdot)$, $\tilde{\Omega}(\cdot)$, and $\tilde{\Theta}(\cdot)$ notations hide poly-logarithmic factors.} time. We design sublinear quantum algorithms for the same task running in $\tilde{O}(\sqrt{n} +\sqrt{d})$ time, a quadratic improvement in both $n$ and $d$. Moreover, our algorithms use the standard quantization of the classical input and generate the same classical output, suggesting minimal overheads when used as subroutines for end-to-end applications. We also demonstrate a tight lower bound (up to poly-log factors) and discuss the possibility of implementation on near-term quantum machines. As a side result, we also give sublinear quantum algorithms for approximating the equilibria of $n$-dimensional matrix zero-sum games with optimal complexity $\tilde{\Theta}(\sqrt{n})$.
\end{abstract}

%%%%%%%%%%%%%%%%%%%%%%%%%%%%%%%%%%%%%%%%%%%%%%%%%%%%%%%%%%%%%%%%%%%%%%%%%%%%%%

\section{Introduction}\label{sec:intro}
\begin{paragraph}{Motivations.}
Classification is a fundamental problem of supervised learning, which takes a training set of data points of known classes as inputs and aims to training a model for predicting the classes of future data points. It is also ubiquitous due to its broad connections and applications to computer vision, natural language processing, statistics, etc.

A fundamental case of classification is \emph{linear classification}, where we are given $n$ data points $\A_{1},\ldots,\A_{n}$ in $\R^{d}$ and a label vector $y\in\{-1,1\}^{n}$. The goal is to find a separating hyperplane, i.e., a unit vector $\x$ in $\R^{d}$, such that
\begin{align}\label{eqn:linear-classification-0}
y_{i}\cdot \A_{i}\trans \x\geq 0\quad\forall\,i\in\range{n}.
\end{align}
By taking $\A_{i}\leftarrow (-1)^{y_{i}}\A_{i}$, it reduces to a \emph{maximin} problem, i.e., $\max_{\x} \min_{i} \A_{i}\trans \x\geq 0$.
The approximation version of linear classification is to find a unit vector
$\bar{\x}\in\R^{d}$ so that
\begin{align}\label{eqn:classification-maximin}
\A_{i}\trans\bar{\x}\geq \max_{\x\in\R_{d}}\min_{i'\in\range{n}}\A_{i'}\trans \x-\epsilon\quad\forall\,i\in\range{n},
\end{align}
i.e., $\bar{\x}$ approximately solves the maximin problem. More generally, we can regard a (nonlinear) classifier as a \emph{kernel-based} classifier by replacing $\A_{i}$ by $\Psi(\A_{i})$ ($\Psi$ being a kernel function). We will focus on algorithms finding approximate classifiers (in the sense of \eqn{classification-maximin}) with \emph{provable guarantees}.

The Perceptron Algorithm for linear classification is one of the oldest algorithms studied in machine learning \cite{novikoff1963convergence,minsky2017perceptrons}, which runs in time $O(nd/\epsilon^{2})$ for finding an $\bar{\x}\in\R^{d}$ satisfying \eqn{classification-maximin}. The state-of-the-art classical result along this line \cite{clarkson2012sublinear} solves linear classification in time $\tilde{O}((n+d)/\epsilon^{2})$. A careful reader might notice that the input to linear classification is $n$ $d$-dimensional vectors with total size $O(nd)$. Hence, the result of \cite{clarkson2012sublinear} is \emph{sub-linear} in its input size. To make it possible, \cite{clarkson2012sublinear} assumes the following entry-wise input model:

\hd{Input model:} given any $i \in [n]$ and $j\in [d]$, the $j$-th entry of $\A_i$ can be recovered in $O(1)$ time.
\vspace{4mm}

The output of \cite{clarkson2012sublinear} is an \emph{efficient} classical representation of $\bar{\x}$ in the sense that every entry of $\bar{\x}$ can be recovered with $\tilde{O}(1)$ cost. It is no surprise that $\bar{\x}$ per se gives such a representation. However, there could be more \emph{succinct and efficient representations} of $\bar{\x}$, which could be reasonable alternatives of $\bar{\x}$ for sub-linear algorithms that run in time less the dimension of $\bar{\x}$ (as we will see in the quantum case). The complexity of \cite{clarkson2012sublinear} is also optimal (up to poly-logarithmic factors) in the above input/output model as shown by the same paper.

Recent developments in quantum computation, especially in the emerging topic of ``quantum machine learning" (see the surveys \cite{biamonte2017quantum,arunachalam2017guest,schuld2015introduction}), suggest that quantum algorithms might offer significant speed-ups for optimization and machine learning problems. In particular, a quantum counterpart of the Perceptron algorithm has been proposed in \cite{kapoor2016quantum} with improved time complexity from $O(nd/\epsilon^2)$ to $\tilde{O}(\sqrt{n}d/\epsilon^{2})$ (details in related works). Motivated both by the significance of classification and the promise of quantum algorithms,  we investigate the \emph{optimal} quantum algorithm for classification. Specifically, we aim to design a quantum counterpart of \cite{clarkson2012sublinear}.

It is natural to require that quantum algorithms make use of the classical input/output model as much as possible to make the comparison fair. In particular, it is favorable to avoid the use of too powerful input data structure which might render any finding of quantum speedup inconclusive, especially in light of a recent development of quantum-inspired classical machine learning algorithms (e.g., \cite{tang2018quantum}). Our choice of input/output models for quantum algorithms is hence almost the same as the classical one, except we allow \emph{coherent} queries to the entries of $\A_i$:

\hd{Quantum input model:} given any $i \in [n]$ and $j\in [d]$, the $j$-th entry of $\A_i$ can be recovered in $O(1)$ time \emph{coherently}.
\vspace{4mm}

Coherent queries allow the quantum algorithm to query many locations in super-position, which is a \emph{standard} assumption that accounts for many quantum speed-ups (e.g., Grover's algorithm \cite{grover1996fast}). A more precise definition is given in \sec{prelim}.

On the other side, our output is exactly the same as classical algorithms, which guarantees no overhead when using our quantum algorithms as subroutines for any applications.
\end{paragraph}

\begin{paragraph}{Contributions.}
Our main contribution is a \emph{tight} characterization (up to poly-log factors) of quantum algorithms for various classification problems in the aforementioned input/output model.

\begin{theorem}[Main theorem]\label{thm:main}
Given $\epsilon=\Theta(1)$, we have quantum algorithms that return an efficient representation of  $\bar{\x}\in\B_{d}$ for the following problems\footnote{Here $\B_{d}$ is the unit ball in $\R^{d}, i.e., \B_{d}:=\big\{a\in\R^{d}\mid \sum_{i\in\range{d}}|a_{i}|^{2}\leq 1\big\}$.}, respectively, with complexity $\tilde{O}(\sqrt{n}+\sqrt{d})$ and high success probability:
\vspace{-2mm}
\begin{itemize}[leftmargin=*]
\item Linear classification (\sec{perceptron}):
\begin{align}
\min_{i\in\range{n}}\A_{i}\trans\bar{\x}\geq \max_{\x\in\B_{d}}\min_{i\in\range{n}}\A_{i}\trans \x-\epsilon.
\end{align}

\item Kernel-based classification (\sec{kernel}):
\begin{align}
\min_{i\in\range{n}}\<\Psi(\A_{i}),\bar{\x}\>\geq\max_{\x\in\B_{d}}\min_{i\in\range{n}}\<\Psi(\A_{i}),\x\>-\epsilon,
\end{align}
where $k(a,b):=\<\Psi(a),\Psi(b)\>$ can be the polynomial kernel $k_{q}(a,b)=(a\trans b)^{q}$ or the Gaussian kernel $k_{\text{Gauss}}(a,b)=\exp(-\|a-b\|^{2})$.

\item Minimum enclosing ball (\sec{MEB}):
\begin{align}
\max_{i\in\range{n}}\|\bar{\x}-\A_{i}\|^{2}\leq\min_{\x\in\R^{d}}\max_{i\in\range{n}}\|\x-\A_{i}\|^{2}+\epsilon.
\end{align}

\item $\ell^{2}$-margin SVM (\sec{L2-SVM}):
\begin{align}\label{eqn:SVM-goal}
\min_{i\in\range{n}}(\A_{i}\trans\bar{\x})^{2}\geq\max\limits_{\x\in\R^{d}}\min_{i\in\range{n}}2\A_{i}\trans \x-\|\x\|^{2}-\epsilon.
\end{align}
\end{itemize}
On the other hand, we show that it requires $\Omega(\sqrt{n}+\sqrt{d})$ queries to the quantum input model to prepare such $\bar{\x}$ for these classification problems (\sec{lower}).
\end{theorem}

Our matching upper and lower bounds $\sqrt{n}+\sqrt{d}$ give a \emph{quadratic} improvement in both $n$ and $d$ comparing to the classical state-of-the-art results in \cite{clarkson2012sublinear}.

Technically, our result is also inspired by the recent development of quantum semidefinite program (SDP) solvers (e.g., \cite{brandao2017quantum}) which provide quantum speed-ups for approximating zero-sum games for the purpose of solving SDPs. Note that such a connection was leveraged classically in another direction in a follow-up work of \cite{clarkson2012sublinear} for solving SDPs \cite{NIPS2011_Garber_Hazan}. However, our algorithm is even simpler because we only use simple quantum state preparation instead of complicated quantum operations in quantum SDP solvers; this is because quantum state preparation is a direct counterpart of the $\ell^{2}$ sampling used in \cite{clarkson2012sublinear} (see \sec{techniques} for details). In a nutshell, our result is a demonstration of quantum speed-ups for sampling-based classical algorithms.

Moreover,  our algorithms are hybrid classical-quantum algorithms where the quantum part is isolated pieces of state preparation connected by classical processing. In addition, special instances of these state preparation might be physically realizable as suggested by some work-in-progress~\cite{Brandao-Talk}. All of the above suggest the possibility of implementing these algorithms on near-term quantum machines~\cite{Preskill2018NISQ}.

In general, we deem our result as a proposal of one end-to-end quantum application in machine learning, with both provable guarantees and the perspective of implementation (at least in prototype) on near-term quantum machines.
\end{paragraph}

\begin{paragraph}{Application to matrix zero-sum games.}
As a side result, our techniques can be applied to solve matrix zero-sum games. To be more specific, the input of the zero-sum game is a matrix $\A\in\R^{n_{1}\times n_{2}}$ and an $\epsilon>0$, and the goal is to find $a\in\R^{n_{1}}$ and $b\in\R^{n_{2}}$ such that\footnote{Here $\Delta_{n}$ is the set of probability distributions on $\range{n}$, i.e., $\Delta_{n}:=\big\{a\in\R^{n}\mid a_{i}\geq 0\ \ \forall\,i\in\range{n}, \sum_{i\in\range{n}}a_{i}=1\big\}$.}
\begin{align}\label{eqn:zero-sum-game-intro}
a^{\dagger}\A b\geq\max_{p\in\Delta_{n_{1}}}\min_{q\in\Delta_{n_{2}}}p^{\dagger}\A q-\epsilon.
\end{align}
If we are given the quantum input model of $A$, we could output such $a$ and $b$ as classical vectors\footnote{In fact, $x$ and $y$ are classical vectors with succinct representations; see more details at \rem{zero-sum}.} with complexity $\tilde{O}(\sqrt{n_{1}+n_{2}}/\epsilon^{4})$ (see \thm{zero-sum-upper}). When $\epsilon=\Theta(1)$, our quantum algorithm is optimal as we prove an $\Omega(\sqrt{n_{1}+n_{2}})$ quantum lower bound (see \thm{zero-sum-lower}).
\end{paragraph}

\begin{paragraph}{Related works.}
We make the following comparisons with existing literatures in quantum machine learning.
\begin{itemize}[leftmargin=*]
\item The most relevant result is the quantum perceptron models in \cite{kapoor2016quantum}. The classical perceptron method \cite{novikoff1963convergence,minsky2017perceptrons} is a pivotal linear classification algorithm. In each iteration, it checks whether \eqn{linear-classification-0} holds; if not, then it searches for a violated constraint $i_{0}$ (i.e., $y_{i_{0}}\A_{i_{0}}\trans\bar{\x}<0$) and update $\bar{\x}\leftarrow\bar{\x}+\A_{i_{0}}$ (up to normalization). This classical perceptron method has complexity $\tilde{O}(nd/\epsilon^{2})$; the quantum counterpart in \cite{kapoor2016quantum} improved the complexity to $\tilde{O}(\sqrt{n}d/\epsilon^{2})$ by applying Grover search \cite{grover1996fast} to find a violated constraint. In contrast, we quantize the sublinear algorithm for linear classification in \cite{clarkson2012sublinear} with techniques inspired by quantum SDP solvers~\cite{brandao2017quantum}. As a result,  we establish a better quantum complexity $\tilde{O}(\sqrt{n}+\sqrt{d})$.

In addition, \cite{kapoor2016quantum} relies on an unusual input model where a data point in $\R^{d}$ is represented by concatenating the the binary representations of the $d$ floating point numbers; if we were only given standard inputs with entry-wise queries to the coordinates of data points, we need a cost of $\Omega(d)$ to transform the data into their input form, giving the total complexity $\tilde{O}(\sqrt{n}d)$.

The same group of authors also gave a quantum algorithm for nearest-neighbor classification with complexity $\tilde{O}(\sqrt{n})$ \cite{wiebe2015quantum}. This complexity also depends on the sparsity of the input data; in the worst case where every data point has $\Theta(d)$ nonzero entries, the complexity becomes $\tilde{O}(\sqrt{n}d^{2})$.

\item There have been rich developments on quantum algorithms for linear algebraic problems. One prominent example is the quantum algorithm for solving linear systems \cite{harrow2009quantum,childs2015quantum}; in particular, they run in time $\poly(\log d)$ for any sparse $d$-dimensional linear systems.
These linear system solvers are subsequently applied to machine learning applications such as cluster assignment \cite{lloyd2013supervised}, support vector machine (SVM) \cite{rebentrost2014QSVM}, etc.

However, these quantum algorithms have two drawbacks. First, they require the input matrix to be \emph{sparse} with efficient access to nonzero elements, i.e., every row/column of the matrix has at most $\poly(\log d)$ nonzero elements and their indexes can be queried in $\poly(\log d)$ time. Second, the outputs of these algorithms are quantum states instead of classical vectors, and it takes $\Omega(d)$ copies of the quantum state to reveal one entry of the output in the worst case. More caveats are listed in \cite{aaronson2015read}.

In contrast, our quantum algorithms do not have the sparsity constraint and work for arbitrary input data, and the outputs of our quantum algorithms are succinct but efficient classical representations of vectors in $\R^{d}$, which can be directly used for classical applications.

\item There are two lines of quantum machine learning algorithms with different input requirements. One of them is based on quantum principal component analysis \cite{lloyd2013quantum} and requires purely quantum inputs.

Another line is the recent development of quantum-inspired classical poly-logarithmic time algorithms for various machine learning tasks such as recommendation systems \cite{tang2018quantum}, principal component analysis \cite{tang2018quantum2}, solving linear systems \cite{chia2018quantum,gilyen2018quantum}, SDPs \cite{chia2018SDP}, and so on. These algorithms follow a Monte-Carlo approach for low-rank matrix approximation \cite{frieze2004fast} and assume the ability to take samples according to the spectral norms of all rows. In other words, these results enforce additional requirements on their input: the input matrix should not only be low-rank but also be preprocessed as the sampling data structure.

\item There are also a few heuristic quantum machine learning approaches for classification \cite{havlicek2018supervised, farhi2018classification, KL18} without theoretical guarantees. We, however, look forward to further experiments based on their proposals.
\end{itemize}
\end{paragraph}

%%%%%%%%%%%%%%%%%%%%%%%%%%%%%%%%%%%%%%%%%%%%%%%%%%%%%%%%%%%%%%%%%%%%%%%%%%%%%%

\section{Preliminaries}\label{sec:prelim}
\begin{paragraph}{Basic notations in quantum computing.}
Quantum mechanics can be formulated in terms of linear algebra. Given any complex Euclidean space $\C^{d}$, we define its computational basis by $\{\vec{e}_{0},\ldots,\vec{e}_{d-1}\}$, where $\vec{e}_{i}=(0,\ldots,1,\ldots,0)\trans $ with the $(i+1)^{\text{th}}$ entry being 1 and other entries being 0. These basic vectors are usually written by \textbf{Dirac notation}: we write $\vec{e}_{i}$ as $|i\>$ (called a ``ket"), and write $\vec{e}_{i}\trans $ as $\<i|$ (called a ``bra").

\emph{Quantum states with dimension $d$} are represented by unit vectors in $\C^{d}$: i.e., a vector $|v\>=(v_{0},\ldots,v_{d-1})\trans $ is a quantum state if $\sum_{i=0}^{d-1}|v_{i}|^{2}=1$. For each $i$, $v_{i}$ is called the \emph{amplitude} in $|i\>$. If there are at least two non-zero amplitudes, quantum state $|v\>$ is in \emph{superposition} of the computational basis, a fundamental feature in quantum mechanics.

\emph{Tensor product} of quantum states is their Kronecker product: if $|u\>\in\C^{d_{1}}$ and $|v\>\in\C^{d_{2}}$, then $|u\>\otimes|v\>\in\C^{d_{1}}\otimes\C^{d_{2}}$ is
\begin{align}
|u\>\otimes|v\>=(u_{0}v_{0},u_{0}v_{1},\ldots,u_{d_{1}-1}v_{d_{2}-1})\trans.
\end{align}
The basic element in classical computers is one bit; similarly, the basic element in quantum computers is one \emph{qubit}, which is a quantum state in $\C^{2}$. Mathematically, a qubit state can be written as $a|0\>+b|1\>$ for some $a,b\in\C$ such that $|a|^{2}+|b|^{2}=1$. An $n$-qubit state can be written as $|v_{1}\>\otimes\cdots\otimes|v_{n}\>$, where each $|v_{i}\>$ ($i\in\range{n}$) is a qubit state; $n$-qubit states are in a Hilbert space of dimension $2^{n}$.

Operations in quantum computation are \emph{unitary transformations} and can be stated in the circuit model\footnote{Uniform circuits have equivalent computational power as Turing machines; however, they are more convenient to use in quantum computation.} where a \emph{$k$-qubit gate} is a unitary matrix in $\C^{2^{k}}$. It is known that two-qubit gates are \emph{universal}, i.e., every $n$-qubit gate can be written as composition of a sequenece of two-qubit gates. Thus, one usually refers to the \emph{number of two-qubit gates} as the \emph{gate complexity} of quantum algorithms.
\end{paragraph}

\begin{paragraph}{Quantum oracle.}
Quantum access to the input data (referred as quantum oracles) needs to be reversible and allows access to different parts of the input data in \emph{superposition} (the essence of quantum speed-ups). Specifically, to access elements in an $n\times d$ matrix $\A$, we exploit an oracle $O_{\A}$ (a unitary on $\C^{n}\otimes \C^{d}\otimes\C^{d_{\text{acc}}}$) such that
\begin{align}\label{eqn:oracle-defn}
O_{\A}(|i\>\otimes|j\>\otimes |z\>)=|i\>\otimes|j\>\otimes|z\oplus \A_{ij}\>
\end{align}
for any $i\in\range{n}$, $j\in\range{d}$ and $z\in\C^{d_{\text{acc}}}$ such that $\A_{ij}$ can be represented in $\C^{d_{\text{acc}}}$. Intuitively, $O_{\A}$ reads the entry $\A_{ij}$ and stores it in the third register. However, to make $O_{\A}$ reversible (and unitary), $O_{\A}$ applies the XOR operation ($\oplus$) on the third register.
Note that $O_{\A}$ is a natural unitary generalization of classical random access to $\A$, or in cases when any entry of $\A$ can be efficiently read.
However, it is potentially stronger when queries become linear combinations of basis vectors, e.g., $\sum_{k} \alpha_{k} |i_{k}\> \otimes |j_{k}\>$.  This is technically how to make superposition of different queries in quantum.

We summarize the quantum notations as follows.
\begin{table}[H]
\centering
\resizebox{0.65\columnwidth}{!}{%
\begin{tabular}{|c||c|c|}
\hline
 & Classical & Quantum \\ \hline\hline
Ket and bra & $\vec{e}_{i}$ and $\vec{e}_{i}\trans$ & $|i\>$ and $\<i|$ \\ \hline
Basis & $\{\vec{e}_{0},\ldots,\vec{e}_{d-1}\}$ & $\{|0\>,\ldots,|d-1\>\}$ \\ \hline
State & $\vec{v}=(v_{0},\ldots,v_{d-1})\trans$ & $|v\>=\sum_{i=0}^{d-1}v_{i}|i\>$ \\ \hline
Tensor & $\vec{u}\otimes\vec{v}$ & $|u\>\otimes|v\>$ or $|u\>|v\>$ \\ \hline
Oracle & $\x=(\A_{ij})_{i,j=1}^{n}$ & $O_{\A}|i\>|j\>|z\>=|i\>|j\>|z\oplus \A_{ij}\>$ \\ \hline
\end{tabular}
}
\caption{Summary of quantum notations used in this paper.}
\label{tab:notations}
\end{table}
\end{paragraph}

\begin{paragraph}{Quantum complexity measure.}
We assume that a single query to the oracle $O_{\A}$ has a unit cost. \emph{Quantum query complexity} is defined as the total counts of oracle queries, and \emph{quantum gate complexity} is defined as the total counts of oracle queries and two-qubit gates.
\end{paragraph}

\begin{paragraph}{Notations.}
Throughout this paper, we denote $\vecc{1}_{n}$ to be the $n$-dimensional all-one vector, and $\A\in\R^{n\times d}$ to be the matrix whose entry in the intersection of its $i^{\text{th}}$ row and $j^{\text{th}}$ column is $\A_{i}(j)$ for all $i\in\range{n}$, $j\in\range{d}$. Without loss of generality, we assume $\A_{1},\ldots,\A_{n}\in\B_{d}$, i.e., all the $n$ data points (also the $n$ rows of $\A$) are normalized to have $\ell^{2}$-norm at most 1.
\end{paragraph}

%%%%%%%%%%%%%%%%%%%%%%%%%%%%%%%%%%%%%%%%%%%%%%%%%%%%%%%%%%%%%%%%%%%%%%%%%%%%%%

\section{Linear classification}\label{sec:perceptron}
\subsection{Techniques}\label{sec:techniques}
At a high level, our quantum algorithm leverages ideas from both classical and quantum algorithm design. We use a primal-dual approach under the multiplicative weight framework \cite{freund1999adaptive}, in particular its improved version in \cite{clarkson2012sublinear} by sampling the update of weight vectors. An important observation of ours is that such classical algorithms can be accelerated significantly in quantum computation, which relies on a seminal technique in quantum algorithm design: amplitude amplification and estimation \cite{grover1996fast,brassard2002quantum}.

\begin{paragraph}{Multiplicative weight under a primal-dual approach.}
Note that linear classification is essentially a minimax problem (zero-sum game); by strong duality, we have
\begin{align}
\sigma=\max_{\x\in\R_{d}}\min_{p\in\Delta_{n}}p\trans \A\x=\min_{p\in\Delta_{n}}\max_{\x\in\R_{d}}p\trans\A\x.
\end{align}
To find its equilibrium point, we adopt an online primal-dual approach with $T$ rounds; at round $t\in\range{T}$, the primal computes $p_{t}\in\Delta_{n}$ and the dual computes $\x_{t}\in\R_{d}$, both based on $p_{\tau}$ and $\x_{\tau}$ for all $\tau\in\range{t-1}$. After $T$ rounds, the average solution $\bar{\x} = \frac{1}{T}\sum_{t=1}^{T}\x_t$ approximately solves the zero-sum game with high probability, i.e., $\min_{p\in\Delta_{n}}p\trans\A\bar{\x}\geq\sigma-\epsilon$.

For the primal problem, we pick $p_{t}$ by the \emph{multiplicative weight} (MW) method. Given a sequence of vectors $r_{1},\ldots,r_{T}\in\R^{n}$, MW sets $w_{1}:=\vecc{1}_{n}$ and for all $t\in\range{T}$, $p_{t}:=w_{t}/\|w_{t}\|_{1}$ and $w_{t+1}(i):=w_{t}(i)f_{\text{w}}(-\eta r_{t}(i))$ for all $i\in\range{n}$, where $f_{\text{w}}$ is a weight function and $\eta$ is the parameter representing the step size. MW promises an upper bound on $\sum_{t=1}^{T}p_{t}\trans r_{t}$, whose precise form depends on the choice of the weight function $f_{\text{w}}$. The most common update is the \emph{exponential weight update}: $f_{\text{w}}(x)=e^{-x}$ \cite{freund1999adaptive}, but in this paper we use a quadratic weight update suggested by \cite{clarkson2012sublinear}, where $w_{t+1}(i):=w_{t}(i)(1-\eta r_{t}(i)+\eta^{2} r_{t}(i)^{2})$. In our primal problem, we set $r_{t}=\A\x_{t}$ for all $t\in\range{T}$ to find $p_{t}$.

For the dual problem, we pick $\x_{t}$ by the \emph{online gradient descent} method \cite{zinkevich2003online}. Given a set of vectors $q_{1},\ldots,q_{T}\in\R^{d}$ such that $\|q_{i}\|_{2}\leq 1$. Let $\x_{0}:=\vecc{0}_{d}$, and $y_{t+1}:=\x_{t}+\frac{1}{\sqrt{T}}q_{t}$, $\x_{t+1}:=\frac{y_{t+1}}{\max\{1,\|y_{t+1}\|\}}$. Then
\begin{align}\label{eqn:online-gradient-descent}
\max_{\x\in\B_{d}}\sum_{t=1}^{T}q_{t}\trans \x-\sum_{t=1}^{T}q_{t}\trans \x_{t}\leq 2\sqrt{T}.
\end{align}
This can be regarded as a \emph{regret} bound, i.e., $\sum_{t=1}^{T}q_{t}\trans \x_{t}$ has at most a regret of $2\sqrt{T}$ compared to the best possible choice of $\x$. In our dual problem, we set $q_{t}$ as a sample of rows of $\A$ following the distribution $p_{t}$.

This primal-dual approach gives a correct algorithm with only $T=\tilde{O}(1/\epsilon^{2})$ iterations. However, the primal step runs in $\Theta(nd)$ time to compute $\A\x_{t}$. To obtain an algorithm that is \emph{sublinear} in the size of $\A$, a key observation by \cite{clarkson2012sublinear} is to replace the precise computation of $\A\x_{t}$ by an unbiased random variable. This is achieved via $\ell^{2}$ sampling of $\x$: we pick $j_t\in\range{d}$ by $j_t=j$ with probability $\x_t(j)^2/\norm{\x_t}^2$, and for all $i\in\range{n}$ we take $\tilde{v}_t(i)=\A_i(j_t)\norm{\x_t}^2/\x_t(j_t)$. The expectation of the random variable $\tilde{v}_t(i)$ satisfies
\begin{align}
\E[\tilde{v}_t(i)]=\sum_{j=1}^{d}\frac{\x_{t}(j)^{2}}{\|\x_{t}\|^{2}}\frac{\A_{i}(j)\|\x_{t}\|^{2}}{\x_{t}(j)}=\A_{i}\x_{t}.
\end{align}
In a nutshell, the update of weight vectors in each iteration need not to be precisely computed because an $\ell^{2}$ sample from $\x$ suffices to promise the provable guarantee of the framework. This trick improves the running time of MW to $O(n)$ and online gradient descent to $O(d)$; since there are $\tilde{O}(1/\epsilon^{2})$ iterations, the total complexity is $\tilde{O}(\frac{n+d}{\epsilon^{2}})$ as claimed in \cite{clarkson2012sublinear}.
\end{paragraph}

\begin{paragraph}{Amplitude amplification and estimation.}
Consider a search problem where we are given a function $f_{\omega}\colon\range{n}\rightarrow\{-1,1\}$ such that $f_{\omega}(i)=1$ iff $i\neq\omega$. To search for $\omega$, classically we need $\Omega(n)$ queries to $f_{\omega}$ as checking all $n$ positions is the only method.

Quantumly, given a unitary $U_{\omega}$ such that $U_{\omega}|i\>=|i\>$ for all $i\neq\omega$ and $U_{\omega}|\omega\>=-|\omega\>$, Grover's algorithm \cite{grover1996fast} finds $\omega$ with complexity $\tilde{O}(\sqrt{n})$. Denote $|s\>=\frac{1}{\sqrt{n}}\sum_{i\in\range{n}}|i\>$ (the uniform superposition), $|s'\>=\frac{1}{\sqrt{n-1}}\sum_{i\in\range{n}/\{\omega\}}|i\>$, and $U_{s}=2|s\>\<s|-I$, the unitary $U_{\omega}$ reflects a state with respect to $|s'\>$ and the unitary $U_{s}$ reflects a state with respect to $|s\>$. If we start with $|s\>$ and denote $\theta=2\arcsin(1/\sqrt{n})$ (the angle between $U_{\omega}|s\>$ and $|s\>$), then the angle between $U_{\omega}|s\>$ and $U_{s}U_{\omega}|s\>$ is \emph{amplified} to $2\theta$, and in general the angle between $U_{\omega}|s\>$ and $(U_{s}U_{\omega})^{k}|s\>$ is $2k\theta$. To find $\omega$, it suffices to take $k=\Theta(\sqrt{n})$ in this quantum algorithm. See \fig{Grover} for an illustration.

\begin{figure}[htbp]
\centering
\includegraphics[width=1.75in]{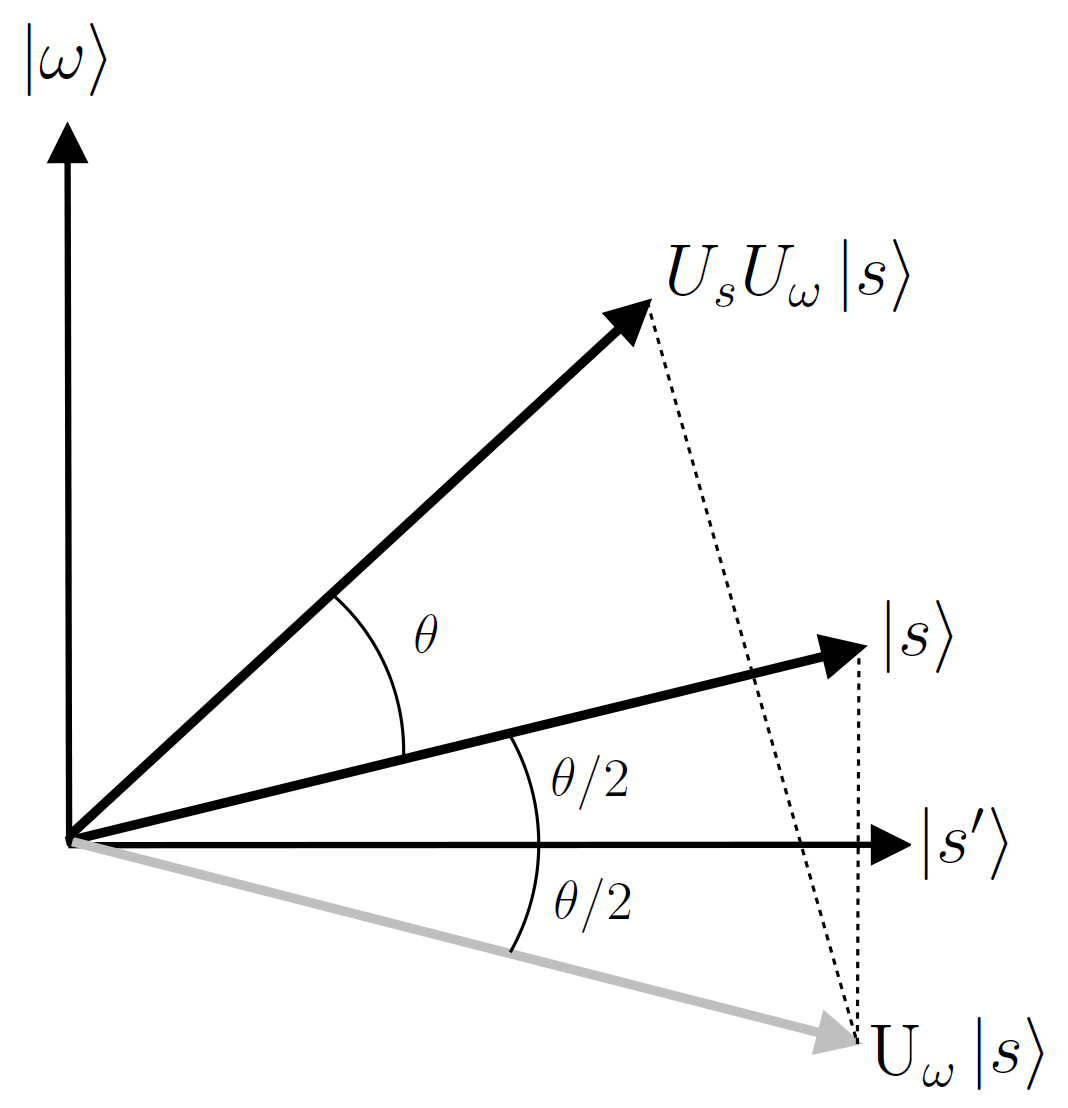}
\caption{Geometric interpretation of Grover's algorithm. This figure is copied from Wikipedia.}
\label{fig:Grover}
\end{figure}
This trick of alternatively applying these two unitaries is called \emph{amplitude amplification}; in general, this provides a quadratic speedup for search problems. For the quantitative version of estimating $\theta$ (not only finding $\omega$), quadratic quantum speedup also holds via an improved version of amplitude amplification called \emph{amplitude estimation} \cite{brassard2002quantum}.

Our \textbf{main technical contribution} is the implementations of amplitude amplification and estimation in the primal-dual approach for solving minimax problems. On the one hand, we achieve quadratic quantum speedup for multiplicative weight update, i.e., we improve the complexity from $\tilde{O}(n)$ to $\tilde{O}(\sqrt{n})$. This is because the $\ell^{2}$ sampling of $\x$ is identical to measuring the quantum state $|\x\>$ in the computational basis; furthermore, we prepare the state $|\x\>$ by amplitude amplification (see \sec{quantum-state-prep}).\footnote{Another common method to prepare quantum states is via quantum random access memory (QRAM). This is incomparable to our approach because preparing the data structure for QRAM takes $\Omega(n)$ cost (though after that one read takes $\tilde{O}(1)$ cost). Here we use amplitude amplification for giving sublinear algorithms. See also \sec{quantum-state-prep}.}

On the other hand, we also achieve quadratic quantum speedup for online gradient descent (improving $\tilde{O}(d)$ to $\tilde{O}(\sqrt{d})$). This is because the main cost of online gradient descent comes from estimating the norms $\|y_{t}\|$, which can be regarded as an amplitude estimation problem; details are given in \sec{quantum-online-gradient}.
\end{paragraph}

\begin{paragraph}{Comparison between classical and quantum results.}
Although our quantum algorithms enjoy quadratic speedups in $n$ and $d$, their executions incur a larger dependence in $\epsilon$: we have worst case $\tilde{O}\big(\frac{\sqrt{n}}{\epsilon^{4}}+\frac{\sqrt{d}}{\epsilon^{8}}\big)$ compared to the classical complexity $\tilde{O}\big(\frac{n}{\epsilon^{2}}+\frac{d}{\epsilon^{2}}\big)$ in \cite{clarkson2012sublinear}.
The main reason of having a larger $\epsilon$-dependence in quantum is because we cannot prepare the weight states in MW via those in previous iterations (i.e., the quantum state $|w_{t}\>$ cannot be prepared by $|w_{t-1}\>$),
and we have to start over every time; this is an intrinsic difficulty due to quantum state preparation.

Therefore, there is a trade-off between \cite{clarkson2012sublinear} and our results for arbitrary $\epsilon$: we provide faster training of the classifiers if we allow a constant error, while the classical algorithms in \cite{clarkson2012sublinear} might work better if we require high-accuracy classifiers.
\end{paragraph}

%=========================================
\subsection{Quantum speedup for multiplicative weights}
First, we give a quantum algorithm for linear classification with complexity $\tilde{O}(\sqrt{n})$:
\begin{theorem}\label{thm:perceptron-quantum}
With success probability at least $2/3$, \algo{perceptron} returns a succinct classical representation of a vector $\bar{\x}\in\R^{d}$ such that
\begin{align}\label{eqn:perceptron-goal}
\A_{i}\bar{\x}\geq \max_{\x\in\B_{d}}\min_{i'\in\range{n}}\A_{i'}\x-\epsilon\quad\forall\,i\in\range{n},
\end{align}
using $\tilde{O}\big(\frac{\sqrt{n}}{\epsilon^{4}}+\frac{d}{\epsilon^{2}}\big)$ quantum gates.
\end{theorem}

\begin{algorithm}[htbp]
\KwInput{$\eps>0$, a quantum oracle $O_{\A}$ for $\A \in \R^{n \times d}$.}
\KwOutput{$\bar{\x}$ that satisfies \eqn{perceptron-goal}.}
Let $T=23^2 \eps^{-2}\log n$, $y_1=\vecc{0}_{d}$, $\eta=\sqrt{\frac{\log n}{T}}$, $u_{1}=\vecc{1}_{n}$, $|p_1\>=\frac{1}{\sqrt{n}}\sum_{i\in\range{n}}|i\>$\;
    \For{$t=1$ \KwTo $T$}{
        Define\footnotemark\,$\x_t:=\frac{y_t}{\max\{1,\norm{y_t}\}}$\; \label{lin:update-x}
        Measure the state $|p_t\>$ in the computational basis and denote the output as $i_t\in\range{n}$\; \label{lin:update-it}
        Define $y_{t+1}:=y_t + \frac{1}{\sqrt{2T}} \A_{i_t}$\; \label{lin:update-y}
        Choose $j_t\in\range{d}$ by $j_t=j$ with probability $\frac{\x_t(j)^2}{\norm{\x_t}^2}$\; \label{lin:update-j}
        For all $i\in\range{n}$, denote $\tilde{v}_t(i)=\A_i(j_t)\frac{\norm{\x_t}^2}{\x_t(j_t)}$, $v_t(i)=\min\{1/\eta,\max\{-1/\eta,\tilde{v}_t(i)\}\}$, and $u_{t+1}(i)=u_t(i) (1-\eta v_t(i) + \eta^2 v_t(i)^2)$. Implement a quantum oracle $O_{t}$ such that for all $i\in\range{n}$, $O_{t}|i\>|0\>=|i\>|u_{t+1}(i)\>$ by \algo{perceptron-quantum-oracle} in \sec{perceptron-quantum-oracle}\; \label{lin:update-oracle-Ot}
        Prepare $|p_{t+1}\>=\frac{1}{\|u_{t+1}\|_{2}}\sum_{i\in\range{n}}u_{t+1}(i)|i\>$ by applying \algo{state-preparation} to $O_{t}$; \label{lin:prepare-pt}
}
Return $\bar{\x} = \frac{1}{T}\sum_{t=1}^{T}\x_t$\;
\caption{Quantum linear classification algorithm.}
\label{algo:perceptron}
\end{algorithm}
\footnotetext{By defining $\x_{t}$ here, we do not write down the whole vector but we construct any query to its entries in $O(1)$ time. For example, the $i^{\text{th}}$ coordinate of $\x_{t}$ is $\x_{t}(i)=\frac{y_t(i)}{\max\{1,\norm{y_t}\}}$, constructed by one query to $y_{t}(i)$. The $y_{t+1}$ in \lin{update-y} is defined in the same sense.\label{fnote:defn-vector}}

Note that \algo{perceptron} is inspired by the classical sublinear algorithm \cite{clarkson2012sublinear} by using online gradient descent in \lin{update-y} and $\ell^{2}$ sampling in \lin{update-j} and \lin{update-oracle-Ot}. However, to achieve the $\tilde{O}(\sqrt{n})$ quantum complexity we use two quantum building blocks: a state preparation procedure in \lin{update-oracle-Ot}, and an oracle implementation procedure in \lin{prepare-pt}; their details are covered in , respectively. The full proof of \thm{perceptron-quantum} is given in \sec{perceptron-quantum-proof}.

%--------------------------------------------------------------------------
\subsubsection{Quantum state preparation with oracles}\label{sec:quantum-state-prep}
We use the following result for quantum state preparation (see, e.g., \cite{grover2000synthesis}):
\begin{proposition}\label{prop:state-prep}
Assume that $a\in\C^{n}$, and we are given a unitary oracle $O_{a}$ such that $O|i\>|0\>=|i\>|a_{i}\>$ for all $i\in\range{n}$. Then \algo{state-preparation} takes $O(\sqrt{n})$ calls to $O_{a}$ for preparing the quantum state $\frac{1}{\|a\|_{2}}\sum_{i\in\range{n}}a_{i}|i\>$ with success probability $1-O(1/n)$.
\end{proposition}

\begin{algorithm}[htbp]
Apply D\"{u}rr-H{\o}yer's algorithm \cite{durr1996quantum} to find $a_{\max}:=\max_{i\in\range{n}}|a_{i}|$ in $O(\sqrt{n})$ time\;
Prepare the uniform superposition $\frac{1}{\sqrt{n}}\sum_{i\in\range{n}}|i\>$\;
Perform the following unitary transformations:
\begin{align}
\frac{1}{\sqrt{n}}\sum_{i\in\range{n}}|i\>
\xmapsto{O_{a}} \frac{1}{\sqrt{n}}\sum_{i\in\range{n}}|i\>|a_{i}\>
&\mapsto \frac{1}{\sqrt{n}}\sum_{i\in\range{n}}|i\>|a_{i}\>\Big(\frac{a_{i}}{a_{\max}}|0\>+\sqrt{1-\frac{|a_{i}|^{2}}{a_{\max}^{2}}}|1\>\Big) \nonumber \\
&\xmapsto{O_{a}^{-1}} \frac{1}{\sqrt{n}}\sum_{i\in\range{n}}|i\>|0\>\Big(\frac{a_{i}}{a_{\max}}|0\>+\sqrt{1-\frac{|a_{i}|^{2}}{a_{\max}^{2}}}|1\>\Big); \label{eqn:state-controlled}
\end{align}
\vspace{-4mm} \BlankLine
Delete the second system in Eq. \eqn{state-controlled}, and rewrite the state as
\begin{align}
\frac{\|a\|_{2}}{\sqrt{n}a_{\max}}\cdot\Big(\frac{1}{\|a\|_{2}}\sum_{i\in\range{n}}a_{i}|i\>\Big)|1\>+|a^{\perp}\>|0\>, \label{eqn:state-controlled-2}
\end{align}
where $|a^{\perp}\>:=\frac{1}{\sqrt{n}}\sum_{i\in\range{n}}\sqrt{1-\frac{|a_{i}|^{2}}{a_{\max}^{2}}}|i\>$ is a garbage state\;
Apply amplitude amplification \cite{brassard2002quantum} for the state in \eqn{state-controlled-2} conditioned on the second system being 1. Return the output\;
\caption{Prepare a pure state given an oracle to its coefficients.}
\label{algo:state-preparation}
\end{algorithm}

Note that the coefficient in \eqn{state-controlled-2} satisfies $\frac{\|a\|_{2}}{\sqrt{n}a_{\max}}\geq\frac{1}{\sqrt{n}}$; therefore, applying amplitude amplification for $O(\sqrt{n})$ times indeed promises that we obtain $|1\>$ on the second system with success probability $1-O(1/n)$, i.e., the state $\frac{1}{\|a\|_{2}}\sum_{i\in\range{n}}a_{i}|i\>$ is prepared in the first system.

\begin{remark}
\algo{state-preparation} is incomparable to state preparation via quantum random access memory (QRAM). QRAM relies on the weak assumption that we start from zero, and every added datum is processed in poly-logarithmic time. In total, this takes at least linear time in the size of the data (see, for instance, \cite{kerenidis2016recommendation}). For the task of \prop{state-prep}, QRAM takes at least $\Omega(n)$ cost.

In this paper, we use the standard model where the input is formulated as an oracle, also widely assumed and used in existing quantum algorithm literatures (e.g., \cite{grover1996fast,harrow2009quantum,childs2015quantum,brandao2017quantum}). Under the standard model, \algo{state-preparation} prepares states with only $O(\sqrt{n})$ cost.

Nevertheless, it is an interesting question to ask whether there is a $\poly(\log(nd))$-time quantum algorithm for linear classification given the existence of a pre-loaded QRAM of $\A$. This would require the ability to take summations of the vectors $\frac{1}{\sqrt{2T}}\A_{i_t}$ in \lin{update-y} of \algo{perceptron} in $\poly(\log(nd))$-time as well as the ability to update the weight state $u_{t+1}$ in \lin{prepare-pt} in $\poly(\log(nd))$-time, both using QRAM. These two tasks are plausible as suggested by classical poly-log time sample-based algorithms for matrix arithmetics under multiplicative weight frameworks \cite{chia2018SDP}, which can potentially be combined with the analysis of QRAM data structures in \cite{kerenidis2016recommendation}; we leave this possibility as an open question.
\end{remark}

%--------------------------------------------------------------------------
\subsubsection{Implementation of the quantum oracle for updating the weight vectors}\label{sec:perceptron-quantum-oracle}
The quantum oracle $O_{t}$ in \lin{update-oracle-Ot} of \algo{perceptron} is implemented by \algo{perceptron-quantum-oracle}. For convenience, we denote $\clip(v,1/\eta):=\min\{1/\eta,\max\{-1/\eta,v\}\}$ for all $v\in\R$.

\begin{algorithm}
\KwInput{$\x_{1},\ldots,\x_{t}\in\R^{d}$, $j_{1},\ldots,j_{t}\in\range{d}$.}
\KwOutput{An oracle $O_{t}$ such that $O_{t}|i\>|0\>=|i\>|u_{t+1}(i)\>$ for all $i\in\range{n}$.}
Define three classical oracles: $O_{s,j}(0)=j_{s}$, $O_{s,\x}(j_{s})=\frac{\|\x_{s}\|^{2}}{\x_{s}(j_{s})}$, and $O_{\text{clip}}(a,b,c)=c\cdot\big(1-\eta\clip(ab,1/\eta)+\eta^{2}\clip(ab,1/\eta)^{2}\big)$\; \label{lin:three-classical-oracles}
\For{$s=1$ \KwTo $t$}{
	Perform the following maps:
	\begin{align}
	|i\>|0\>|0\>|0\>|u_{s}(i)\>&\xmapsto{O_{s,j}}|i\>|j_{s}\>|0\>|0\>|u_{s}(i)\> \label{eqn:perceptron-quantum-oracle-1} \\
		&\xmapsto{O_{\A}}|i\>|j_{s}\>|\A_{i}(j_{s})\>|0\>|u_{s}(i)\> \label{eqn:perceptron-quantum-oracle-2} \\
		&\xmapsto{O_{s,\x}}|i\>|j_{s}\>|\A_{i}(j_{s})\>\Big|\frac{\|\x_{s}\|^{2}}{\x_{s}(j_{s})}\Big\>|u_{s}(i)\> \label{eqn:perceptron-quantum-oracle-3} \\
		&\xmapsto{O_{\text{clip}}}|i\>|j_{s}\>|\A_{i}(j_{s})\>\Big|\frac{\|\x_{s}\|^{2}}{\x_{s}(j_{s})}\Big\>|u_{s+1}(i)\> \label{eqn:perceptron-quantum-oracle-4} \\
		&\xmapsto{O_{s,\x}^{-1}}|i\>|j_{s}\>|\A_{i}(j_{s})\>|0\>|u_{s+1}(i)\> \label{eqn:perceptron-quantum-oracle-5} \\
		&\xmapsto{O_{\A}^{-1}}|i\>|j_{s}\>|0\>|0\>|u_{s+1}(i)\> \label{eqn:perceptron-quantum-oracle-6} \\
		&\xmapsto{O_{s,j}^{-1}}|i\>|0\>|0\>|0\>|u_{s+1}(i)\>. \label{eqn:perceptron-quantum-oracle-7}
	\end{align}
}
\caption{Quantum oracle for updating the weight state.}
\label{algo:perceptron-quantum-oracle}
\end{algorithm}

Because we have stored $\x_{s}$ and $j_{s}$, we could construct classical oracles $O_{s,j}(0)=j_{s}$, $O_{s,\x}(j_{s})=\frac{\|\x_{s}\|^{2}}{\x_{s}(j_{s})}$ with $O(1)$ complexity. In the algorithm, we first call $O_{s,j}$ to compute $j_{s}$ and store it into the second register in \eqn{perceptron-quantum-oracle-1}. In \eqn{perceptron-quantum-oracle-2}, we call the quantum oracle $O_{\A}$ for the value $\A_{i}(j_{s})$, which is stored into the third register. In \eqn{perceptron-quantum-oracle-3}, we call $O_{s,\x}$ to compute $\frac{\|\x_{s}\|^{2}}{\x_{s}(j_{s})}$ and store it into the fourth register. In \eqn{perceptron-quantum-oracle-4}, because we have $\A_{i}(j_{s})$ and $\frac{\|\x_{s}\|^{2}}{\x_{s}(j_{s})}$ at hand, we could use $\tilde{O}(1)$ arithmetic computations to compute $\tilde{v}_s(i)=\A_i(j_s)\norm{\x_s}^2/\x_t(j_s)$ and
\begin{align}
u_{s+1}(i)=u_s(i)\big(1-\eta \clip(\tilde{v}_s(i), 1/\eta) + \eta^2 \clip(\tilde{v}_s(i), 1/\eta)^2\big).
\end{align}
We then store $u_{s+1}(i)$ into the fifth register. In \eqn{perceptron-quantum-oracle-5}, \eqn{perceptron-quantum-oracle-6}, and \eqn{perceptron-quantum-oracle-7}, we uncompute the steps in \eqn{perceptron-quantum-oracle-3}, \eqn{perceptron-quantum-oracle-2}, and \eqn{perceptron-quantum-oracle-1}, respectively (we need these steps in \algo{perceptron-quantum-oracle} to keep its unitarity).

In total, between \eqn{perceptron-quantum-oracle-1}-\eqn{perceptron-quantum-oracle-7} we use 2 queries to $O_{\A}$ and $\tilde{O}(1)$ additional arithmetic computations. Because $s$ goes from 1 to $t$, in total we use $2t$ queries to $O_{\A}$ and $\tilde{O}(t)$ additional arithmetic computations.

%--------------------------------------------------------------------------
\subsubsection{Proof of \thm{perceptron-quantum}}\label{sec:perceptron-quantum-proof}
To prove \thm{perceptron-quantum}, we use the following five lemmas proved in \cite{clarkson2012sublinear} for analyzing the online gradient gradient descent and $\ell^{2}$ sampling outcomes:

\begin{lemma}[Lemma A.2 of \cite{clarkson2012sublinear}]\label{lem:A-2}
The updates of $\x$ in \lin{update-x} and $y$ in \lin{update-y} satisfy
\begin{align}
\max_{\x\in\B_{n}}\sum_{t\in\range{T}}\A_{i_{t}}\x\leq\sum_{t\in\range{T}}\A_{i_{t}}\x_{t}+2\sqrt{2T}.
\end{align}
\end{lemma}

\begin{lemma}[Lemma 2.3 of \cite{clarkson2012sublinear}]\label{lem:2-3}
For any $t\in\range{T}$, denote $p_{t}$ to be the unit vector in $\R^{n}$ such that $(p_{t})_{i}=|\<i|p_{t}\>|^{2}$ for all $i\in\range{n}$. Then the update for $p_{t+1}$ in \lin{prepare-pt} satisfies
\begin{align}
\sum_{t\in\range{T}}p_{t}\trans v_{t}\leq\min_{i\in\range{n}}\sum_{t\in\range{T}}v_{t}(i)+\eta\sum_{t\in\range{T}}p_{t}\trans v_{t}^{2}+\frac{\log n}{\eta},
\end{align}
where $v_{t}^{2}$ is defined as $(v_{t}^{2})_{i}:=(v_{t})_{i}^{2}$ for all $i\in\range{n}$.
\end{lemma}

\begin{lemma}[Lemma 2.4 of \cite{clarkson2012sublinear}]\label{lem:2-4}
With probability at least $1-O(1/n)$,
\begin{align}
\max_{i\in\range{n}}\sum_{t\in\range{T}}\big[v_{t}(i)-\A_{i}\x_{t}\big]\leq 4\eta T.
\end{align}
\end{lemma}

\begin{lemma}[Lemma 2.5 of \cite{clarkson2012sublinear}]\label{lem:2-5}
With probability at least $1-O(1/n)$,
\begin{align}
\Big|\sum_{t\in\range{T}}\A_{i_{t}}\x_{t}-\sum_{t\in\range{T}}p_{t}\trans v_{t}\Big|\leq 10\eta T.
\end{align}
\end{lemma}

\begin{lemma}[Lemma 2.6 of \cite{clarkson2012sublinear}]\label{lem:2-6}
With probability at least $3/4$,
\begin{align}
\sum_{t\in\range{T}}p_{t}\trans v_{t}^{2}\leq 8T.
\end{align}
\end{lemma}

\begin{proof}
We first prove the correctness of \algo{perceptron}. By \lem{A-2}, we have
\begin{align}\label{eqn:perceptron-1}
\sum_{t\in\range{T}}\A_{i_{t}}\x_{t}\geq\max_{\x\in\B_{n}}\sum_{t\in\range{T}}\A_{i_{t}}\x-2\sqrt{2T}\geq T\sigma-2\sqrt{2T}.
\end{align}
On the other hand, \lem{2-4} implies that for any $i\in\range{n}$,
\begin{align}
\sum_{t\in\range{T}}\A_{i}\x_{t}\geq\sum_{t\in\range{T}}v_{t}(i)-4\eta T.
\end{align}
Together with \lem{2-3}, we have
\begin{align}\label{eqn:perceptron-2}
\sum_{t\in\range{T}}p_{t}\trans v_{t}\leq\min_{i\in\range{n}}\sum_{t\in\range{T}}\A_{i}\x_{t}+\eta\sum_{t\in\range{T}}p_{t}\trans v_{t}^{2}+\frac{\log n}{\eta}+4\eta T.
\end{align}
Plugging \lem{2-5}, \lem{2-6}, and \eqn{perceptron-1} into \eqn{perceptron-2}, with probability at least $\frac{3}{4}-2\cdot O(\frac{1}{n})\geq\frac{2}{3}$,
\begin{align}
\min_{i\in\range{n}}\sum_{t\in\range{T}}\A_{i}\x_{t}\geq -\frac{\log n}{\eta}-8\eta T-4\eta T+T\sigma-2\sqrt{2T}-10\eta T\geq T\sigma-22\eta T-\frac{\log n}{\eta}.
\end{align}
Since $T=23^2 \eps^{-2}\log n$ and $\eta=\sqrt{\frac{\log n}{T}}$, we have
\begin{align}
\min_{i\in\range{n}}\A_{i}\bar{\x}=\frac{1}{T}\min_{i\in\range{n}}\sum_{t=1}^{T}\A_{i}\x_{t}\geq\sigma-23\sqrt{\frac{\log n}{T}}\geq\sigma-\epsilon
\end{align}
with probability at least $2/3$, which is exactly \eqn{perceptron-goal}.
\\\\
Now we analyze the gate complexity of \algo{perceptron}. To run \lin{update-x} and \lin{update-y}, we need $d$ time and space to compute and store $\x_{t}$ and $y_{t+1}$; for all $t\in\range{T}$, this takes total complexity $O(dT)=\tilde{O}(\frac{d}{\epsilon^{2}})$. It takes another $O(dT)=\tilde{O}(\frac{d}{\epsilon^{2}})$ cost to compute $j_{t}$ for all $t\in\range{T}$ in \lin{update-j}.

The quantum part of \algo{perceptron} mainly happens at \lin{update-oracle-Ot} and \lin{prepare-pt}, where we prepare the quantum state $|p_{t+1}\>$ instead of computing the coefficients $u_{t+1}(i)$ one by one for all $i\in\range{n}$. To be more specific, we construct an oracle $O_{t}$ such that $O_{t}|i\>|0\>=|i\>|u_{t+1}(i)\>$ for all $i\in\range{n}$. This is achieved iteratively, i.e., at iteration $s$ we map $|i\>|u_{s}(i)\>$ to $|i\>|u_{s+1}(i)\>$. The full details are given in \algo{perceptron-quantum-oracle} in \sec{perceptron-quantum-oracle}; in total, one query to $O_{t}$ is implemented by $2t$ queries to $O_{\A}$ and $\tilde{O}(t)$ additional arithmetic computations.

Finally, we prepare the state $|p_{t+1}\>=\frac{1}{\|u_{t+1}\|_{2}}\cdot\sum_{i\in\range{n}}u_{t+1}(i)|i\>$ in \lin{prepare-pt} using $O(\sqrt{n})$ calls to $O_{t}$, which are equivalent to $O(\sqrt{n}t)$ calls to $O_{\A}$ by \lin{update-oracle-Ot} and $\tilde{O}(\sqrt{n}t)$ additional arithmetic computations. Therefore, the total complexity of \lin{prepare-pt} for all $t\in\range{T}$ is
\begin{align}
\sum_{t=1}^{T}\tilde{O}(\sqrt{n}t)=\tilde{O}(\sqrt{n}T^{2})=\tilde{O}\Big(\frac{\sqrt{n}}{\epsilon^{4}}\Big).
\end{align}
In all, the total complexity of \algo{perceptron} is $\tilde{O}\big(\frac{\sqrt{n}}{\epsilon^{4}}+\frac{d}{\epsilon^{2}}\big)$, establishing our statement.

Finally, the output $\bar{\x}$ has a succinct classical representation with space complexity $O(\log n/\epsilon^{2})$. To achieve this, we save $2T=O(\log n/\epsilon^{2})$ values in \algo{perceptron}: $i_{1},\ldots,i_{T}$ and $\|y_{1}\|,\ldots,\|y_{T}\|$; it then only takes $O(\log n/\epsilon^{2})$ cost to recover any coordinate of $\bar{\x}$ by \lin{update-x} and \lin{update-y}.
\end{proof}

\begin{remark}\label{rem:perceptron-PAC}
\thm{perceptron-quantum} could also be applied to the PAC model. For the case where there exists a hyperplane classifying all data points correctly with margin $\sigma$, and assume that the margin is not small in the sense that $\frac{1}{\sigma^{2}}<d$, PAC learning theory implies that the number of examples needed for training a classifier of error $\delta$ is $O(1/\sigma^{2}\delta)$. As a result, we have a quantum algorithm that computes a $\sigma/2$-approximation to the best classifier with cost
\begin{align}
\tilde{O}\Big(\frac{\sqrt{1/\sigma^{2}\delta}}{\sigma^{4}}+\frac{d}{\sigma^{2}}\Big)=\tilde{O}\Big(\frac{1}{\sigma^{5}\sqrt{\delta}}+\frac{d}{\sigma^{2}}\Big).
\end{align}
This is better than the classical complexity $O(\frac{1}{\sigma^{4}\delta}+\frac{d}{\sigma^{2}})$ in \cite{clarkson2012sublinear} as long as $\delta\leq\sigma^{2}$, which is plausible under the assumption that the margin $\sigma$ is large.
\end{remark}

%=========================================
\subsection{Quantum speedup for online gradient descent}\label{sec:quantum-online-gradient}
\paragraph{Norm estimation by amplitude estimation.}
We further improve the dependence in $d$ to $\tilde{O}(\sqrt{d})$. To achieve this, we cannot update $\x_{t}$ and $y_{t}$ in \lin{update-x} and \lin{update-y} by each coordinate because storing $\x_{t}$ or $y_{t}$ would already take cost at least $d$. We solve this issue by not updating $\x_{t}$ and $y_{t}$ explicitly and instead only computing $\|y_{t}\|$ for all $i\in\range{T}$. This norm estimation is achieved by the following lemma:
\begin{lemma}\label{lem:norm-estimation}
Assume that $F\colon\range{d}\to[0,1]$ with a quantum oracle $O_{F}|i\>|0\>=|i\>|F(i)\>$ for all $i\in\range{d}$. Denote $m=\frac{1}{d}\sum_{i=1}^{d}F(i)$. Then for any $\delta>0$, there is a quantum algorithm that uses $O(\sqrt{d}/\delta)$ queries to $O_{F}$ and returns an $\tilde{m}$ such that $|\tilde{m}-m|\leq\delta m$ with probability at least 2/3.
\end{lemma}

Our proof of \lem{norm-estimation} is based on amplitude estimation:
\begin{theorem}[Theorem 15 of \cite{brassard2002quantum}]\label{thm:quantum-counting}
For any $0<\epsilon<1$ and Boolean function $f\colon\range{d}\rightarrow\{0,1\}$ with quantum oracle $O_{f}|i\>|0\>=|i\>|f(i)\>$ for all $i\in\range{d}$, there is a quantum algorithm that outputs an estimate $\hat{t}$ to $t=|f^{-1}(1)|$ such that
\begin{align}
|\hat{t}-t|\leq\epsilon t
\end{align}
with probability at least $8/\pi^{2}$, using $O(\frac{1}{\epsilon}\sqrt{\frac{d}{t}})$ evaluations of $O_{f}$. If $t=0$, the algorithm outputs $\hat{t}=0$ with certainty and $O_{f}$ is evaluated $O(\sqrt{d})$ times.
\end{theorem}

\begin{proof}
Assume that $F(i)$ has $l$ bits for precision for all $i\in\range{d}$ (in our paper, we take $l=O(1)$, say $l=64$ for double float precision), and for all $k\in\range{l}$ denote $F_{k}(i)$ as the $k^{\text{th}}$ bit of $F(i)$; denote $n_{k}=\sum_{i\in\range{d}}F_{k}(i)$.

We apply \thm{quantum-counting} to all the $l$ bits of $n_{k}$ using $O(\sqrt{d}/\delta)$ queries (taking $\epsilon=\delta/2$), which gives an approximation $\hat{n}_{k}$ of $n_{k}$ such that with probability at least $8/\pi^{2}$ we have $|n_{k}-\hat{n}_{k}|\leq\delta n_{k}/2$ if $n_{k}\geq 1$, and $\hat{n}_{k}=0$ if $n_{k}=0$. Running this procedure for $\Theta(\log l)$ times and take the median of all returned $\hat{n}_{k}$, and do this for all $k\in\range{l}$, Chernoff's bound promises that with probability $2/3$ we have
\begin{align}
|n_{k}-\hat{n}_{k}|\leq \delta n_{k}\quad\forall\,k\in\range{l}.
\end{align}
As a result, if we take $\tilde{m}=\frac{1}{d}\sum_{k\in\range{l}}\frac{\hat{n}_{k}}{2^{k}}$, and observe that $m=\frac{1}{d}\sum_{k\in\range{l}}\frac{n_{k}}{2^{k}}$, with probability at least $2/3$ we have
\begin{align}
|\tilde{m}-m|\leq \frac{1}{d}\sum_{k\in\range{l}}\Big|\frac{\hat{n}_{k}}{2^{k}}-\frac{n_{k}}{2^{k}}\Big|\leq \frac{1}{d}\sum_{k\in\range{l}}\frac{\delta n_{k}}{2^{k}}=\delta m.
\end{align}
The total quantum query complexity is $O(l\log l\cdot\sqrt{d}/\delta)=O(\sqrt{d}/\delta)$.
\end{proof}

\paragraph{Quantum algorithm with $\tilde{O}(\sqrt{d})$ cost.}
Instead of updating $y_{t}$ explicitly in \lin{update-y} of \algo{perceptron}, we save the $i_{t}$ for all $t\in\range{T}$ in \lin{update-it}, which only takes $\tilde{O}(1/\epsilon^{2})$ cost in total but we can directly generate $y_{t}$ given $i_{1},\ldots,i_{t}$. Furthermore, notice that the probabilities in the $\ell^{2}$ sampling in \lin{update-j} do not change because $\x_t(j)^2/\norm{\x_t}^2=y_t(j)^2/\norm{y_t}^2$; it suffices to replace $\tilde{v}_t(i)=\A_i(j_t)\norm{\x_t}^2/\x_t(j_t)$ by $\tilde{v}_t(i)=\A_i(j_t)\norm{y_t}^2/(y_t(j_t)\max\{1,\|y_{t}\|\})$ in \lin{update-oracle-Ot}. These observations result in \algo{perceptron-d} with the following result:

\begin{theorem}\label{thm:perceptron-quantum-d}
With success probability at least $2/3$, there is a quantum algorithm that returns a succinct classical representation of a vector $\bar{\x}\in\R^{d}$ such that
\begin{align}\label{eqn:perceptron-goal-d}
\A_{i}\bar{\x}\geq \max_{\x\in\B_{d}}\min_{i'\in\range{n}}\A_{i'}\x-\epsilon\quad\forall\,i\in\range{n},
\end{align}
using $\tilde{O}\big(\frac{\sqrt{n}}{\epsilon^{4}}+\frac{\sqrt{d}}{\epsilon^{8}}\big)$ quantum gates.
\end{theorem}

\begin{algorithm}
\KwInput{$\eps>0$, a quantum oracle $O_{\A}$ for $\A \in \R^{n \times d}$.}
\KwOutput{$\bar{\x}$ that satisfies \eqn{perceptron-goal}.}
Let $T=27^2 \eps^{-2}\log n$, $y_1=\vecc{0}_{d}$, $\eta=\sqrt{\frac{\log n}{T}}$, $u_{1}=\vecc{1}_{n}$, $|p_1\>=\frac{1}{\sqrt{n}}\sum_{i\in\range{n}}|i\>$\;
    \For{$t=1$ \KwTo $T$}{
        Measure the state $|p_t\>$ in the computational basis and denote the output as $i_t\in\range{n}$\; \label{lin:update-it-d}
        Define\footnotemark\,$y_{t+1}:=y_t + \frac{1}{\sqrt{2T}} \A_{i_t}$\; \label{lin:update-y-d}
        Apply \lem{norm-estimation} for $2\lceil\log T\rceil$ times to estimate $\|y_{t}\|^{2}$ with precision $\delta=\eta^{2}$, and take the median of all the $2\lceil\log T\rceil$ outputs, denoted $\widetilde{\|y_{t}\|}^{2}$\; \label{lin:estimate-norm-d}
        Choose $j_t\in\range{d}$ by $j_t=j$ with probability $y_t(j)^2/\norm{y_t}^2$, which is achieved by applying \algo{state-preparation} to prepare the quantum state $|y_{t}\>$ and measure in the computational basis\; \label{lin:update-j-d}
        For all $i\in\range{n}$, denote $\tilde{v}_t(i)=\A_i(j_t)\widetilde{\|y_{t}\|}^2/\big(y_t(j_t)\max\{1,\widetilde{\|y_{t}\|}\}\big)$, $v_t(i)=\clip(\tilde{v}_t(i), 1/\eta)$, and $u_{t+1}(i)=u_t(i) (1-\eta v_t(i) + \eta^2 v_t(i)^2)$. Apply \algo{perceptron-quantum-oracle} to prepare an oracle $O_{t}$ such that $O_{t}|i\>|0\>=|i\>|u_{t+1}(i)\>$ for all $i\in\range{n}$, using $2t$ queries to $O_{\A}$ and $\tilde{O}(t)$ additional arithmetic computations\; \label{lin:update-oracle-Ot-d}
        Prepare the state $|p_{t+1}\>=\frac{1}{\|u_{t+1}\|_{2}}\sum_{i\in\range{n}}u_{t+1}(i)|i\>$ using \algo{state-preparation} and $O_{t}$; \label{lin:prepare-pt-d}
}
Return $\bar{\x} = \frac{1}{T}\sum_{t=1}^{T}\frac{y_t}{\max\{1,\widetilde{\norm{y_t}}\}}$\;
\caption{Quantum linear classification algorithm with $\tilde{O}(\sqrt{d})$ cost.}
\label{algo:perceptron-d}
\end{algorithm}
\footnotetext{The meaning of the definition here is the same as \fnote{defn-vector} in \algo{perceptron}.}

\begin{proof}
For clarification, we denote
\begin{align}
\tilde{v}_{t,\text{approx}}(i)=\frac{\A_i(j_t)\widetilde{\|y_{t}\|}^2}{y_t(j_t)\max\{1,\widetilde{\|y_{t}\|}\}}, \qquad \tilde{v}_{t,\text{true}}(i)=\frac{\A_i(j_t)\|y_{t}\|^2}{y_t(j_t)\max\{1,\|y_{t}\|\}}\quad \forall\,i\in\range{n}.
\end{align}
In other words, the $\tilde{v}_{t}$ in \lin{update-oracle-Ot-d} of \algo{perceptron-d} is $\tilde{v}_{t,\text{approx}}$, an approximation of $\tilde{v}_{t,\text{true}}$. We prove:
\begin{align}\label{eqn:perceptron-d-1}
|\tilde{v}_{t,\text{approx}}(i)-\tilde{v}_{t,\text{true}}(i)|\leq\eta\quad \forall\,i\in\range{n}.
\end{align}
Without loss generality, we can assume that $\tilde{v}_{t,\text{true}}(i),\tilde{v}_{t,\text{approx}}(i)\leq 1/\eta$; otherwise, they are both truncated to $1/\eta$ by the clip function in \lin{update-oracle-Ot-d} and no error occurs. For convenience, we denote $m=\|y_{t}\|^2$ and $\tilde{m}=\widetilde{\|y_{t}\|}^2$. Then
\begin{align}\label{eqn:perceptron-d-2}
|\tilde{v}_{t,\text{approx}}(i)-\tilde{v}_{t,\text{true}}(i)|=\tilde{v}_{t,\text{true}}(i)\cdot\Big|\frac{\tilde{v}_{t,\text{approx}}(i)}{\tilde{v}_{t,\text{true}}(i)}-1\Big|\leq\frac{1}{\eta}\cdot\Big|\frac{\tilde{v}_{t,\text{approx}}(i)}{\tilde{v}_{t,\text{true}}(i)}-1\Big|.
\end{align}
When $\|y_{t}\|\geq 1$ we have $\frac{\tilde{v}_{t,\text{approx}}(i)}{\tilde{v}_{t,\text{true}}(i)}=\frac{\tilde{m}}{m}$; when $\|y_{t}\|\leq 1$ we have $\frac{\tilde{v}_{t,\text{approx}}(i)}{\tilde{v}_{t,\text{true}}(i)}=\sqrt{\frac{\tilde{m}}{m}}$. Because in \lin{estimate-norm-d} $\widetilde{\|y_{t}\|}^{2}$ is the median of $2\lceil\log T\rceil$ executions of \lem{norm-estimation}, with failure probability at most $1-(2/3)^{2\log T}=O(1/T^{2})$ we have $|\frac{\tilde{m}}{m}-1|\leq\delta$; given there are $T$ iterations in total, the probability that \lin{estimate-norm-d} always succeeds is at least $1-T\cdot O(1/T^{2})=1-o(1)$, and we have
\begin{align}
\Big|\frac{\tilde{m}}{m}-1\Big|,\ \Big|\sqrt{\frac{\tilde{m}}{m}}-1\Big|\leq\delta.
\end{align}
Plugging this into \eqn{perceptron-d-2}, we have
\begin{align}\label{eqn:perceptron-d-3}
|\tilde{v}_{t,\text{approx}}(i)-\tilde{v}_{t,\text{true}}(i)|\leq\frac{\delta}{\eta}=\eta,
\end{align}
which proves \eqn{perceptron-d-1}.

Now we prove the correctness of \algo{perceptron-d}. By \eqn{perceptron-d-1} and \lem{2-4}, with probability at least $1-O(1/n)$ we have
\begin{align}
\max_{i\in\range{n}}\sum_{t\in\range{T}}\big[v_{t}(i)-\A_{i}\x_{t}\big]\leq 4\eta T+\eta T=5\eta T,
\end{align}
where $\x_{t}=\frac{y_t}{\max\{1,\widetilde{\norm{y_t}}\}}$ for all $t\in\range{T}$. By \eqn{perceptron-d-1} and \lem{2-5}, with probability at least $1-O(1/n)$ we have
\begin{align}
\Big|\sum_{t\in\range{T}}\A_{i_{t}}\x_{t}-\sum_{t\in\range{T}}p_{t}\trans v_{t}\Big|\leq 10\eta T+\eta T=11\eta T;
\end{align}
by \eqn{perceptron-d-1} and \lem{2-6}, with probability at least $3/4$ we have
\begin{align}
\sum_{t\in\range{T}}p_{t}\trans v_{t}^{2}\leq 8T+2T=10T.
\end{align}
As a result, similar to the proof of \thm{perceptron-quantum}, we have
\begin{align}
\min_{i\in\range{n}}\sum_{t\in\range{T}}\A_{i}\x_{t}\geq -\frac{\log n}{\eta}-10\eta T-5\eta T+T\sigma-2\sqrt{2T}-11\eta T\geq T\sigma-26\eta T-\frac{\log n}{\eta}.
\end{align}
Since $T=27^2 \eps^{-2}\log n$ and $\eta=\sqrt{\frac{\log n}{T}}$, we have
\begin{align}
\min_{i\in\range{n}}\A_{i}\bar{\x}=\frac{1}{T}\min_{i\in\range{n}}\sum_{t=1}^{T}\A_{i}\x_{t}\geq\sigma-27\sqrt{\frac{\log n}{T}}\geq\sigma-\epsilon
\end{align}
with probability at least $2/3$, which is exactly \eqn{perceptron-goal-d}.

It remains to analyze the time complexity. Same as the proof of \thm{perceptron-quantum}, the complexity in $n$ is $\tilde{O}(\frac{\sqrt{n}}{\epsilon^4})$. It remains to show that the complexity in $d$ is $\tilde{O}(\frac{\sqrt{n}}{\epsilon^8})$. The cost in $d$ in \algo{perceptron} and \algo{perceptron-d} differs at \lin{estimate-norm-d} and \lin{update-j-d}. We first look at \lin{estimate-norm-d}; because
\begin{align}
y_{t}=\frac{1}{\sqrt{2T}}\sum_{\tau=1}^{T}\A_{i_{\tau}},
\end{align}
one query to a coefficient of $y_{t}$ takes $t=\tilde{O}(1/\epsilon^{2})$ queries to $O_{\A}$. Next, since $\A_{i}\in\B_{n}$ for all $i\in\range{n}$, we know that $\A_{ij}\in[-1,1]$ for all $i\in\range{n}$, $j\in\range{d}$; to apply \lem{norm-estimation} ($F$ should have image domain in $[0,1]$) we need to renormalize $y_{t}$ by a factor of $t=\tilde{O}(1/\epsilon^{2})$. In addition, notice that $\delta=\eta^{2}=\Theta(\epsilon^{2})$; as a result, the query complexity of executing \lem{norm-estimation} is $\tilde{O}(\sqrt{d}/\epsilon^{2})$. Finally, there are in total $T=\tilde{O}(1/\epsilon^{2})$ iterations. Therefore, the total complexity in \lin{estimate-norm-d} is
\begin{align}\label{eqn:estimate-norm-d-Line5}
\tilde{O}\Big(\frac{1}{\epsilon^{2}}\Big)\cdot \tilde{O}\Big(\frac{1}{\epsilon^{2}}\Big)\cdot \tilde{O}\Big(\frac{\sqrt{d}}{\epsilon^{2}}\Big)\cdot \tilde{O}\Big(\frac{1}{\epsilon^{2}}\Big)=\tilde{O}\Big(\frac{\sqrt{d}}{\epsilon^{8}}\Big).
\end{align}

Regarding the complexity in $d$ in \lin{update-j-d}, the cost is to prepare the pure state $|y_{t}\>$ whose coefficient is proportional to $y_{t}$. To achieve this, we need $t=\tilde{O}(1/\epsilon^{2})$ queries to $O_{\A}$ (for summing up the rows $\A_{i_{1}},\ldots,\A_{i_{t}}$) such that we have an oracle $O_{y_{t}}$ satisfying $O_{y_{t}}|j\>|0\>=|j\>|y_{t}(j)\>$ for all $j\in\range{d}$. By \algo{state-preparation}, the query complexity of preparing $|y_{t}\>$ using $O_{y_{t}}$ is $O(\sqrt{d})$. Because there are in total $T=\tilde{O}(1/\epsilon^{2})$ iterations, the total complexity in \lin{update-j-d} is
\begin{align}\label{eqn:estimate-norm-d-Line6}
\tilde{O}\Big(\frac{1}{\epsilon^{2}}\Big)\cdot O(\sqrt{d})\cdot \tilde{O}\Big(\frac{1}{\epsilon^{2}}\Big)=\tilde{O}\Big(\frac{\sqrt{d}}{\epsilon^{4}}\Big).
\end{align}

In all, the total complexity in $d$ is $\tilde{O}(\sqrt{d}/\epsilon^{8})$ as dominated by \eqn{estimate-norm-d-Line5}. Finally, $\bar{\x}$ has a succinct classical representation: using $i_{1},\ldots,i_{T}$ obtained from \lin{update-it-d} and $\widetilde{\|y_{1}\|}^{2},\ldots,\widetilde{\|y_{T}\|}^{2}$ obtained from \lin{estimate-norm-d}, we could restore a coordinate of $\bar{\x}$ in time $T=\tilde{O}(1/\epsilon^{2})$.
\end{proof}

\begin{remark}
For practical applications of linear classification, typically the number of data points $n$ is larger than the dimension $d$, so in practice \thm{perceptron-quantum} might perform better than \thm{perceptron-quantum-d}. Nevertheless, the $\tilde{O}(\sqrt{d})$ complexity in \thm{perceptron-quantum-d} matches our quantum lower bound (see \thm{perceptron-quantum-lower}).
\end{remark}

%%%%%%%%%%%%%%%%%%%%%%%%%%%%%%%%%%%%%%%%%%%%%%%%%%%%%%%%%%%%%%%%%%%%%%%%%%%%%%

\section{Applications}\label{sec:applications}
As introduced in \sec{techniques}, the $\ell^{2}$ sampling of $\x$ picks $j_t\in\range{d}$ by $j_t =j$ with probability $\x(j)^2/\norm{\x}^2$, and the expectation of the random variable $\A_i(j_t)\norm{\x}^2/\x(j_t)$ is $\A_{i}\x$. Here, if we consider some alternate random variables, we could give unbiased estimators of nonlinear functions of $\A$. We first look at the general case of applying kernel functions \cite{scholkopf2002learning} in \sec{kernel}. We then look at the special case of quadratic problems in \sec{quadratic-problems} as they enjoy simple forms that can be applied to finding minimum enclosing balls \cite{saha2011new} and $\ell^{2}$-margin support vector machines \cite{suykens1999least}. Finally, we follow this methodology to give a sublinear quantum algorithm for solving matrix zero-sum games in \sec{zero-sum}.

%=========================================
\subsection{Kernel methods}\label{sec:kernel}
Having quantum algorithms for solving linear classification at hand, it is natural to consider linear classification under kernels. Let $\Psi\colon\R^{d}\mapsto\mathcal{H}$ be a mapping into a reproducing kernel Hilbert space (RKHS), and the problem is to find the classifier $h\in\mathcal{H}$ that solves the maximin problem
\begin{align}\label{eqn:classification-kernel}
\sigma=\max_{h\in\mathcal{H}}\min_{i\in\range{n}}\<h,\Psi(\A_{i})\>,
\end{align}
where the kernel is defined as $k(a,b):=\<\Psi(a),\Psi(b)\>$ for all $a,b\in\R^{d}$.

Classically, \cite{clarkson2012sublinear} gave the following result for classification under efficiently-computable kernels, following the linear classification algorithm therein:
\begin{theorem}[Lemma 5.3 of \cite{clarkson2012sublinear}]\label{thm:kernel-CHW}
Denote $T_{k}$ as the time cost for computing $k(\A_{i},\A_{j})$ for some $i,j\in\range{n}$, and denote $L_{k}$ as the time cost for computing a random variable $\tilde{k}(\A_{i},\A_{j})$ for some $i,j\in\range{n}$ such that $\E[\tilde{k}(\A_{i},\A_{j})]=k(\A_{i},\A_{j})$ and $\var[k(\A_{i},\A_{j})]\leq 1$. Then there is a classical algorithm that runs in time
\begin{align}\label{eqn:classification-kernel-classical-time}
\tilde{O}\Big(\frac{L_{k}n+d}{\epsilon^{2}}+\min\Big\{\frac{T_{k}}{\epsilon^{4}},\frac{L_{k}}{\epsilon^{6}}\Big\}\Big)
\end{align}
and returns a vector $\bar{h}\in\mathcal{H}$ such that with high success probability $\<\bar{h},\Psi(\A_{i})\>\geq\sigma-\epsilon$ for all $i\in\range{n}$.
\end{theorem}

Quantumly, we give an algorithm for classification under kernels based on \algo{perceptron-d}:
\begin{algorithm}
\KwInput{$\eps>0$, a quantum oracle $O_{\A}$ for $\A \in \R^{n \times d}$.}
\KwOutput{$\bar{\x}$ that satisfies \eqn{perceptron-goal}.}
Let $T=27^2 \eps^{-2}\log n$, $y_1=\vecc{0}_{d}$, $\eta=\sqrt{\frac{\log n}{T}}$, $u_{1}=\vecc{1}_{n}$, $|p_1\>=\frac{1}{\sqrt{n}}\sum_{i\in\range{n}}|i\>$\;
    \For{$t=1$ \KwTo $T$}{
        Measure the state $|p_t\>$ in the computational basis and denote the output as $i_t\in\range{n}$\; \label{lin:update-it-kernel}
        Define $y_{t+1}:=y_t + \frac{1}{\sqrt{2T}} \Psi(\A_{i_t})$\; \label{lin:update-y-kernel}
        Apply \lem{norm-estimation} for $2\lceil\log T\rceil$ times to estimate $\|y_{t}\|^{2}$ with precision $\delta=\eta^{2}$, and take the median of all the $2\lceil\log T\rceil$ outputs, denoted $\widetilde{\|y_{t}\|}^{2}$\; \label{lin:estimate-norm-kernel}
        Choose $j_t\in\range{d}$ by $j_t=j$ with probability $y_t(j)^2/\norm{y_t}^2$, which is achieved by applying \algo{state-preparation} to prepare the quantum state $|y_{t}\>$ and measure in the computational basis\; \label{lin:update-j-kernel}
        For all $i\in\range{n}$, denote $\tilde{v}_t(i)=\Psi(\A_i)(j_t)\widetilde{\|y_{t}\|}^2/\big(y_t(j_t)\max\{1,\widetilde{\|y_{t}\|}\}\big)$, $v_t(i)=\clip(\tilde{v}_t(i), 1/\eta)$, and $u_{t+1}(i)=u_t(i) (1-\eta v_t(i) + \eta^2 v_t(i)^2)$. Apply \algo{perceptron-quantum-oracle} to prepare an oracle $O_{t}$ such that $O_{t}|i\>|0\>=|i\>|u_{t+1}(i)\>$ for all $i\in\range{n}$, using $2t$ queries to $O_{\A}$ and $\tilde{O}(t)$ additional arithmetic computations\; \label{lin:update-oracle-Ot-kernel}
        Prepare the state $|p_{t+1}\>=\frac{1}{\|u_{t+1}\|_{2}}\sum_{i\in\range{n}}u_{t+1}(i)|i\>$ using \algo{state-preparation} and $O_{t}$; \label{lin:prepare-pt-kernel}
}
Return $\bar{\x} = \frac{1}{T}\sum_{t=1}^{T}\frac{y_t}{\max\{1,\widetilde{\norm{y_t}}\}}$\;
\caption{Quantum algorithm for kernel-based classification.}
\label{algo:perceptron-kernel}
\end{algorithm}

\thm{perceptron-quantum-d} and \thm{kernel-CHW} imply that our quantum kernel-based classifier has time complexity
\begin{align}\label{eqn:classification-kernel-quantum-time}
\tilde{O}\Big(\frac{L_{k}\sqrt{n}}{\epsilon^{4}}+\frac{\sqrt{d}}{\epsilon^{8}}+\min\Big\{\frac{T_{k}}{\epsilon^{4}},\frac{L_{k}}{\epsilon^{6}}\Big\}\Big).
\end{align}

For polynomial kernels of degree $q$, i.e., $k_{q}(x,y)=(x\trans y)^{q}$, we have $L_{k_{q}}=q$ by taking the product of $q$ independent $\ell^{2}$ samples (this is an unbiased estimator of $(x\trans y)^{q}$ and the variance of each sample is at most 1). As a result of \eqn{classification-kernel-quantum-time},
\begin{corollary}\label{cor:kernel-polynomial}
For the polynomial kernel of degree $q$, there is a quantum algorithm that solves the classification task within precision $\epsilon$ with gate complexity $\tilde{O}\big(\frac{q\sqrt{n}}{\epsilon^{4}}+\frac{q\sqrt{d}}{\epsilon^{8}}\big)$.
\end{corollary}
\noindent
Compared to the classical complexity $\tilde{O}\big(\frac{q(n+d)}{\epsilon^{2}}+\min\big\{\frac{d\log q}{\epsilon^{4}},\frac{q}{\epsilon^{6}}\big\}\big)$ in Corollary 5.4 of \cite{clarkson2012sublinear}, our quantum algorithm gives quadratic speedups in $n$ and $d$.

For Gaussian kernels, i.e., $k_{\text{Gauss}}(x,y)=\exp(-\|x-y\|^{2})$, Corollary 5.5 of \cite{clarkson2012sublinear} proved that $L_{k_{\text{Gauss}}}=1/s^{4}$ if the Gaussian has standard deviation $s$. As a result,
\begin{corollary}\label{cor:kernel-Gaussian}
For the polynomial kernel of degree $q$, there is a quantum algorithm that solves the classification task within precision $\epsilon$ with gate complexity $\tilde{O}\big(\frac{\sqrt{n}}{s^{4}\epsilon^{4}}+\frac{\sqrt{d}}{s^{4}\epsilon^{8}}\big)$.
\end{corollary}
\noindent
This still gives quadratic speedups in $n$ and $d$ compared to the classical complexity $\tilde{O}\big(\frac{n+d}{s^{4}\epsilon^{2}}+\min\big\{\frac{d}{\epsilon^{4}},\frac{1}{s^{4}\epsilon^{6}}\big\}\big)$ in Corollary 5.5 of \cite{clarkson2012sublinear}.

%=========================================
\subsection{Quadratic machine learning problems}\label{sec:quadratic-problems}
We consider the maximin problem of a quadratic function:
\begin{align}\label{eqn:quadratic-maximin}
\max_{\x\in\R^{d}}\min_{p\in\Delta_{n}}p\trans (b+2\A\x-\vecc{1}_{n}\|\x\|^{2})=\max_{\x\in\R^{d}}\min_{i\in\range{n}}b_{i}+2\A_{i}\x-\|\x\|^{2},
\end{align}
where $b\in\R^{n}$ and $\A\in\R^{n\times d}$. Note that the function $b_{i}+2\A_{i}\x-\|\x\|^{2}$ in Eq. \eqn{quadratic-maximin} is 2-strongly convex; as a result, the regret of the online gradient descent after $T$ rounds can be improved to $O(\log T)$ by \cite{sra2012optimization} instead of $O(\sqrt{T})$ as in Eq. \eqn{online-gradient-descent}. In addition, $\ell^{2}$ sampling of the $\x$ in \algo{perceptron} and \algo{perceptron-d} still works: consider the random variable $\x=b_{i}+\frac{2\A_{i}(j)\|\x\|^{2}}{\x(j)}-\|\x\|^{2}$ where $j=k$ with probability $\frac{\x(k)^{2}}{\|\x\|^{2}}$. Then the expectation of $\x$ is
\begin{align}
\E[X]=\sum_{j=1}^{d}\frac{\x(j)^{2}}{\|\x\|^{2}}\Big(b_{i}+\frac{2\A_{i}(j)\|\x\|^{2}}{\x(j)}-\|\x\|^{2}\Big)=b_{i}+2\A_{i}\x-\|\x\|^{2},
\end{align}
i.e., $\x$ is an unbiased estimator of the quadratic form in \eqn{quadratic-maximin}. As a result, given the quantum oracle $O_{\A}$ in \eqn{oracle-defn}, we could give sublinear quantum algorithms for such problems; these include two important problems: minimum enclosing balls (MEB) and $\ell^{2}$-margin supper vector machines (SVM).

%--------------------------------------------------------------------------
\subsubsection{Minimum enclosing ball}\label{sec:MEB}
In the minimum enclosing ball (MEB) problem we have $b_{i}=-\|\A_{i}\|^{2}$ for all $i\in\range{n}$; Eq. \eqn{quadratic-maximin} then becomes $\max_{\x\in\R^{d}}\min_{i\in\range{n}}-\|\A_{i}\|^{2}+2\A_{i}\x-\|\x\|^{2}=-\min_{\x\in\R^{d}}\max_{i\in\range{n}}\|\x-\A_{i}\|^{2}$, which is the smallest radius of the balls that contain all the $n$ data points $\A_{1},\ldots,\A_{n}$.

Denote $\sigma_{\MEB}=\min_{\x\in\R^{d}}\max_{i\in\range{n}}\|\x-\A_{i}\|^{2}$, we have:

\begin{theorem}\label{thm:MEB-quantum}
There is a quantum algorithm that returns a vector $\bar{\x}\in\R^{d}$ such that with probability at least $2/3$,
\begin{align}\label{eqn:MEB-goal}
\max_{i\in\range{n}}\|\bar{\x}-\A_{i}\|^{2}\leq\sigma_{\MEB}+\epsilon,
\end{align}
using $\tilde{O}\big(\frac{\sqrt{n}}{\epsilon^{4}}+\frac{d}{\epsilon}\big)$ quantum gates; the quantum gate complexity can also be improved to $\tilde{O}\big(\frac{\sqrt{n}}{\epsilon^{4}}+\frac{\sqrt{d}}{\epsilon^{7}}\big)$.
\end{theorem}

We omit the proof of \thm{MEB-quantum} because it directly follows from \thm{perceptron-quantum} (see also Theorem 3.1 in \cite{clarkson2012sublinear}) and \thm{perceptron-quantum-d}. For the $\tilde{O}(\sqrt{d})$ complexity result, the same idea of \algo{perceptron-d} is applied to estimate the norm $\|y_{t}\|$ by amplitude estimation; the error dependence becomes $1/\epsilon^{7}$ because with high probability, the number of iterations that we obtain a new $y_{t}$ in \lin{update-y-d} is $O(\alpha T)=\tilde{O}(1/\epsilon)$, and the other overheads in $\epsilon$ is still $\tilde{O}(\sqrt{d}/\epsilon^{6})$ (see Eq. \eqn{estimate-norm-d-Line5}).

%--------------------------------------------------------------------------
\subsubsection{$\ell^{2}$-margin SVM}\label{sec:L2-SVM}
To estimate the margin of a support vector machine (SVM) in $\ell^{2}$-norm, we take $b_{i}=0$ for all $i\in\range{n}$; Eq. \eqn{quadratic-maximin} then becomes solving $\sigma_{\SVM}:=\max\limits_{\x\in\R^{d}}\min\limits_{i\in\range{n}}2\A_{i}\x-\|\x\|^{2}$.

Notice that $\sigma_{\SVM}\geq 0$ because $2\A_{i}\x-\|\x\|^{2}=0$ for all $i\in\range{n}$ when $\x=0$. For the case $\sigma_{\SVM}>0$ and taking $0<\epsilon<\sigma_{\SVM}$, similar to \thm{MEB-quantum} we have:
\begin{corollary}\label{cor:SVM-quantum}
There is a quantum algorithm that returns a vector $\bar{\x}\in\R^{d}$ such that with probability at least $2/3$,
\begin{align}\label{eqn:SVM-goal}
\min_{i\in\range{n}}2\A_{i}\bar{\x}-\|\bar{\x}\|^{2}\geq\sigma_{\SVM}-\epsilon>0,
\end{align}
using $\tilde{O}\big(\frac{\sqrt{n}}{\epsilon^{4}}+\frac{d}{\epsilon}\big)$ quantum gates; the quantum gate complexity can also be improved to $\tilde{O}\big(\frac{\sqrt{n}}{\epsilon^{4}}+\frac{\sqrt{d}}{\epsilon^{7}}\big)$.
\end{corollary}

Note that \eqn{SVM-goal} implies that $\A_{i}\bar{\x}>0$ for all $i\in\range{n}$; furthermore, by the AM-GM inequality we have $\frac{(\A_{i}\bar{\x})^{2}}{\|\bar{\x}\|^{2}}+\|\bar{\x}\|^{2}\geq 2\A_{i}\bar{\x}$, and hence
\begin{align}
\min_{i\in\range{n}}\Big(\frac{\A_{i}\bar{\x}}{\|\bar{\x}\|}\Big)^{2}\geq\min_{i\in\range{n}}2\A_{i}\bar{\x}-\|\bar{\x}\|^{2}\geq\sigma_{\SVM}-\epsilon.
\end{align}
If we denote $\hat{\x}=\bar{\x}/\|\bar{\x}\|$, then $\A_{i}\hat{\x}\geq\sqrt{\sigma_{\SVM}-\epsilon}>0$ for all $i\in\range{n}$. Consequently, if the data $\A$ is from an SVM, we obtain a normalized direction $\hat{\x}$ (in $\ell^{2}$-norm) such that all data points have a margin of at least $\sqrt{\sigma_{\SVM}-\epsilon}$. Classically, this task takes time $\tilde{O}(n+d)$ for constant $\sigma_{\SVM}$ by \cite{clarkson2012sublinear}, but our quantum algorithm only takes time $\tilde{O}(\sqrt{n}+\sqrt{d})$.

%=========================================
\subsection{Matrix zero-sum games}\label{sec:zero-sum}
Our $\ell^{2}$-sampling technique can also be adapted to solve matrix zero-sum games as an application. To be more specific, the input of a zero-sum game is a matrix $\A\in\R^{n_{1}\times n_{2}}$, and the goal is to find $a\in\R^{n_{1}}$ and $b\in\R^{n_{2}}$ such that
\begin{align}\label{eqn:zero-sum-game}
a^{\dagger}\A b\geq\max_{p\in\Delta_{n_{1}}}\min_{q\in\Delta_{n_{2}}}p^{\dagger}\A q - \epsilon
\end{align}
for some $\epsilon>0$; such $(a,b)$ is called an \emph{$\epsilon$-optimal strategy}. It is shown in \cite[Proposition 1]{grigoriadis1995sublinear} that for $0<\epsilon<0.1$, an $\epsilon$-optimal strategy for the $(n_{1}+n_{2}+1)$-dimensional anti-symmetric matrix
\begin{align}
\A'=\left( \begin{array}{ccc}
0 & \A & -\vecc{1}_{n_{1}} \\
-\A^{\dagger} & 0 & \vecc{1}_{n_{2}} \\
\vecc{1}_{n_{1}} & -\vecc{1}_{n_{2}} & 0
\end{array} \right)
\end{align}
implies an $18\epsilon$-optimal strategy for $\A$. Therefore, without loss of generality, we could assume that $\A$ is an $n$-dimensional anti-symmetric matrix (by taking $n=n_{1}+n_{2}+1$). In this case, the game value $\max_{p\in\Delta_{n}}\min_{q\in\Delta_{n}}p^{\dagger}\A q$ in \eqn{zero-sum-game} equals to 0, and due to symmetry finding an $\epsilon$-optimal strategy reduces to find an $\x\in\Delta_{n}$ such that
\begin{align}\label{eqn:zero-sum-game-x}
\A\x\leq\epsilon\cdot\vecc{1}_{n},
\end{align}
where $\leq$ applies to each coordinate. As a normalization, we assume that $\max_{i,j\in\range{n}}|\A_{i,j}|\leq 1$.

Classically, one query to $\A$ is to ask for one entry in the matrix, whereas quantumly we assume the oracle in \eqn{oracle-defn}. Inspired by Ref. \cite[Theorem 1]{grigoriadis1995sublinear}, we give the following result for solving the zero-sum game:
\begin{theorem}\label{thm:zero-sum-upper}
With success probability at least $2/3$, \algo{zero-sum-game} returns a vector $\bar{\x}\in\R^{n}$ such that $\A\bar{\x}\leq\epsilon\cdot\vecc{1}_{n}$, using $\tilde{O}\big(\frac{\sqrt{n}}{\epsilon^{4}}\big)$ quantum gates.
\end{theorem}

\begin{algorithm}
\KwInput{$\eps>0$, a quantum oracle $O_{\A}$ for $\A \in \R^{n \times n}$.}
\KwOutput{$\bar{\x}\in\Delta_{n}$ that satisfies \eqn{zero-sum-game-x}.}
Let $T \gets 4\eps^{-2}\log n$, $A \gets \vecc{0}_{n}$, $|p_1\> \gets \frac{1}{\sqrt{n}}\sum_{i\in\range{n}}|i\>$\;
    \For{$t=1$ \KwTo $T$}{
        Measure the state $|p_t\>$ in the computational basis and denote the output as $k_t\in\range{n}$\; \label{lin:zero-sum-game-k}
        Update the $k_t$-\text{th} coordinate of $A$: $A_{k_{t}} \gets A_{k_{t}}+1$\; \label{lin:zero-sum-game-X}
        Prepare the state
        \begin{align}\label{eqn:zero-sum-game-p}
        |p_{t+1}\>=\frac{\sum_{i\in\range{n}}\exp[\epsilon\sum_{\tau=1}^{t}\A_{i,k_{\tau}}/4]|i\>}{\sqrt{\sum_{j\in\range{n}}\exp[\epsilon\sum_{\tau=1}^{t}\A_{j,k_{\tau}}/2]}}
        \end{align}
        using \algo{state-preparation}\; \label{lin:zero-sum-game-p}
}
Return $\bar{\x} = A/T$\;
\caption{Sublinear quantum algorithm for solving zero-sum games.}
\label{algo:zero-sum-game}
\end{algorithm}

\begin{proof}
We first prove the correctness of \algo{zero-sum-game}. We denote $P_{i}(t):=\exp[\epsilon\sum_{\tau=1}^{t}\A_{i,k_{\tau}}/2]$ and $p_{i}(t)=P_{i}(t)/\sum_{j=1}^{n}P_{j}(t)$ for all $i\in\range{n}$ and $t\in\range{T}$. Then $|p_{t+1}\>=\sum_{i=1}^{n}\sqrt{p_{i}(t)}|i\>$. We also denote the potential function $\Phi(t)=\sum_{i=1}^{n}P_{i}(t)$. It satisfies
\begin{align}\label{eqn:zero-sum-game-1}
\Phi(t)=\sum_{i=1}^{n}P_{i}(t)=\sum_{i=1}^{n}P_{i}(t-1)\exp[\epsilon \A_{i,k_{t}}/2]=\Phi(t-1)\sum_{i=1}^{n}p_{i}(t-1)\exp[\epsilon \A_{i,k_{t}}/2].
\end{align}
Since \lin{zero-sum-game-k} selects $k_{t}$ with probability $p_{k_{t}}(t-1)$, Eq. \eqn{zero-sum-game-1} implies
\begin{align}\label{eqn:zero-sum-game-2}
\E[\Phi(t)]=\Phi(t-1)\sum_{i,k=1}^{n}p_{i}(t-1)p_{k}(t-1)\exp[\epsilon \A_{i,k}/2].
\end{align}
Because $|\A_{i,k}|\leq 1$, we have
\begin{align}\label{eqn:zero-sum-game-3}
\exp[\epsilon \A_{i,k}/2]\leq 1-\frac{\epsilon \A_{i,k}}{2}+\frac{\epsilon^{2}}{6}.
\end{align}
Also because $\A$ is skew-symmetric, we have $\sum_{i,k=1}^{n}p_{i}(t-1)p_{k}(t-1)\A_{i,k}=0$. Plugging this and \eqn{zero-sum-game-3} into \eqn{zero-sum-game-2}, we have $\E[\Phi(t)]\leq \E[\Phi(t-1)]\big(1+\frac{\epsilon^{2}}{6}\big)$. As a result of induction,
\begin{align}\label{eqn:zero-sum-game-5}
\E[\Phi(T)]\leq \Phi(0)\Big(1+\frac{\epsilon^{2}}{6}\Big)^{T}\leq n\exp[T\epsilon^{2}/6]\leq n^{5/3}.
\end{align}
By Markov's inequality, we have $\Phi(T)\leq 3n^{5/3}\leq n^{2}$ with probability at least $2/3$. Notice that $\Phi(T)\leq n^{2}$ implies $P_{i}(T)\leq n^{2}$ for all $i\in\range{n}$, i.e., $\epsilon\sum_{\tau=1}^{T}\A_{i,k_{\tau}}/2\leq 2\ln n$ for all $i\in\range{n}$. The $i$-th coordinate of $\A\bar{\x}$ satisfies
\begin{align}
(\A\bar{\x})_{i}=\frac{1}{T}(\A A)_{i}=\frac{1}{T}\sum_{\tau=1}^{T}\A_{i,k_{\tau}}\leq\frac{\epsilon^{2}}{4\ln n}\cdot\frac{4\ln n}{\epsilon}=\epsilon;
\end{align}
since this is true for all $i\in\range{n}$, we have $\A\bar{\x}\leq\epsilon\cdot\vecc{1}_{n}$.

It remains to prove the complexity claim. The measurement in \lin{zero-sum-game-k} takes $O(\log n)$ gates, and the update in \lin{zero-sum-game-X} also takes $O(\log n)$ gates because it only adds 1 to one of the $n$ coordinates. The complexity of the algorithm thus mainly comes from \lin{zero-sum-game-p} for preparing $|p_{t+1}\>$. Notice that $\arg\max\exp[\epsilon\sum_{\tau=1}^{t}\A_{i,k_{\tau}}/4]=\arg\max\sum_{\tau=1}^{t}\A_{i,k_{\tau}}$, which can be computed in $O(t\sqrt{n})$ queries to the oracle $O_{\A}$. Similarly, the amplitude amplification in \algo{state-preparation} can also be done with cost $O(t\sqrt{n})$. In total, the time complexity of \algo{zero-sum-game} is
\begin{align}
\sum_{t=1}^{T}O(t\sqrt{n})=O(T^{2}\sqrt{n})=\tilde{O}\Big(\frac{\sqrt{n}}{\epsilon^{4}}\Big).
\end{align}
\end{proof}

\begin{remark}\label{rem:zero-sum}
The output of \algo{zero-sum-game} is a classical vector in $\Delta_{n}$; furthermore, it has a succinct representation of $O(\log^{2}n/\epsilon^{2})$ bits: \lin{zero-sum-game-X} in each iteration add 1 to one of the $n$ coordinates and hence can be stored in $\lceil\log_{2}n\rceil$ bits, and there are in total $O(\log n/\epsilon^{2})$ rounds. Therefore, such output can be directly useful for classical applications, which distinguishes from many quantum machine learning algorithms that output a quantum state (whose applications are more subtle).
\end{remark}

%%%%%%%%%%%%%%%%%%%%%%%%%%%%%%%%%%%%%%%%%%%%%%%%%%%%%%%%%%%%%%%%%%%%%%%%%%%%%%

\section{Quantum lower bounds}\label{sec:lower}
All quantum algorithms (upper bounds) above have matching lower bounds in $n$ and $d$. Assuming $\epsilon=\Theta(1)$ and given the oracle $O_{\A}$ in \eqn{oracle-defn}, we prove quantum lower bounds on linear classification, minimum enclosing ball, and matrix zero-sum games in \sec{linear-lower}, \sec{MEB-lower}, and \sec{zero-sum-lower}, respectively.

%=========================================
\subsection{Linear classification}\label{sec:linear-lower}
Recall that the input of the linear classification problem is a matrix $\A\in\R^{n\times d}$ such that $\A_{i}\in\B_{d}$ for all $i\in\range{n}$ ($\A_{i}$ being the $i^{\text{th}}$ row of $\A$), and the goal is to approximately solve
\begin{align}\label{eqn:linear-classification-rewrite}
\sigma:=\max_{\x\in\B_{d}}\min_{p\in\Delta_{n}}p\trans\A\x=\max_{\x\in\B_{d}}\min_{i\in\range{n}}\A_{i}\x.
\end{align}
Given the quantum oracle $O_{\A}$ such that $O_{\A}|i\>|j\>|0\>=|i\>|j\>|\A_{ij}\>\ \forall\,i\in\range{n}, j\in\range{d}$, \thm{perceptron-quantum-d} solves this task with high success probability with cost $\tilde{O}\big(\frac{\sqrt{n}}{\epsilon^{4}}+\frac{\sqrt{d}}{\epsilon^{8}}\big)$. We prove a quantum lower bound that matches this upper bound in $n$ and $d$ for constant $\epsilon$:
\begin{theorem}\label{thm:perceptron-quantum-lower}
Assume $0<\epsilon<0.04$. Then to return an $\bar{\x}\in\B_{d}$ satisfying
\begin{align}\label{eqn:perceptron-goal-rewrite}
\A_{j}\bar{\x}\geq \max_{\x\in\B_{d}}\min_{i\in\range{n}}\A_{i}\x-\epsilon\quad\forall\,j\in\range{n}
\end{align}
with probability at least $2/3$, we need $\Omega(\sqrt{n}+\sqrt{d})$ quantum queries to $O_{\A}$.
\end{theorem}

\begin{proof}
Assume we are given the promise that $\A$ is from one of the two cases below:
\begin{enumerate}
\item There exists an $l\in\{2,\ldots,d\}$ such that $\A_{11}=-\frac{1}{\sqrt{2}}$, $\A_{1l}=\frac{1}{\sqrt{2}}$; $\A_{21}=\A_{2l}=\frac{1}{\sqrt{2}}$; there exists a unique $k\in\{3,\ldots,n\}$ such that $\A_{k1}=1$, $\A_{kl}=0$; $\A_{ij}=\frac{1}{\sqrt{2}}$ for all $i\in\{3,\ldots,n\}/\{k\}$, $j\in\{1,l\}$, and $\A_{ij}=0$ for all $i\in\range{n}$, $j\notin\{1,l\}$.
\item There exists an $l\in\{2,\ldots,d\}$ such that $\A_{11}=-\frac{1}{\sqrt{2}}$, $\A_{1l}=\frac{1}{\sqrt{2}}$; $\A_{21}=\A_{2l}=\frac{1}{\sqrt{2}}$; $\A_{ij}=\frac{1}{\sqrt{2}}$ for all $i\in\{3,\ldots,n\}$, $j\in\{1,l\}$, and $\A_{ij}=0$ for all $i\in\range{n}$, $j\notin\{1,l\}$.
\end{enumerate}
Notice that the only difference between these two cases is a row where the first entry is 1 and the $l^{\text{th}}$ entry is 0; they have the following pictures, respectively:
\begin{align}
\text{Case 1: }\A&=\left( \begin{array}{cccccccc}
        -\frac{1}{\sqrt{2}} & 0 & \cdots & 0 & \frac{1}{\sqrt{2}} & 0 & \cdots & 0 \\
        \frac{1}{\sqrt{2}} & 0 & \cdots & 0 & \frac{1}{\sqrt{2}} & 0 & \cdots & 0 \\
        \vdots & \vdots & \ddots & \vdots & \vdots & \vdots & \ddots & \vdots \\
        \frac{1}{\sqrt{2}} & 0 & \cdots & 0 & \frac{1}{\sqrt{2}} & 0 & \cdots & 0 \\
        1 & 0 & \cdots & 0 & 0 & 0 & \cdots & 0 \\
        \frac{1}{\sqrt{2}} & 0 & \cdots & 0 & \frac{1}{\sqrt{2}} & 0 & \cdots & 0 \\
        \vdots & \vdots & \ddots & \vdots & \vdots & \vdots & \ddots & \vdots \\
        \frac{1}{\sqrt{2}} & 0 & \cdots & 0 & \frac{1}{\sqrt{2}} & 0 & \cdots & 0
      \end{array} \right); \label{eqn:perceptron-lower-construction-1}
\end{align}
and
\begin{align}
\text{Case 2: }\A&=\left( \begin{array}{cccccccc}
        -\frac{1}{\sqrt{2}} & 0 & \cdots & 0 & \frac{1}{\sqrt{2}} & 0 & \cdots & 0 \\
        \frac{1}{\sqrt{2}} & 0 & \cdots & 0 & \frac{1}{\sqrt{2}} & 0 & \cdots & 0 \\
        \vdots & \vdots & \ddots & \vdots & \vdots & \vdots & \ddots & \vdots \\
        \frac{1}{\sqrt{2}} & 0 & \cdots & 0 & \frac{1}{\sqrt{2}} & 0 & \cdots & 0 \\
        \vdots & \vdots & \ddots & \vdots & \vdots & \vdots & \ddots & \vdots \\
        \frac{1}{\sqrt{2}} & 0 & \cdots & 0 & \frac{1}{\sqrt{2}} & 0 & \cdots & 0
      \end{array} \right). \label{eqn:perceptron-lower-construction-2}
\end{align}

We denote the maximin value in \eqn{linear-classification-rewrite} of these cases as $\sigma_{1}$ and $\sigma_{2}$, respectively. We have:

\begin{itemize}[leftmargin=*]
\item $\sigma_{2}=\frac{1}{\sqrt{2}}$.
\end{itemize}

On the one hand, consider $\bar{\x}=\vec{e}_{l}\in\B_{d}$ (the vector in $\R^{d}$ with the $l^{\text{th}}$ coordinate being 1 and all other coordinates being 0). Then $\A_{i}\bar{\x}=\frac{1}{\sqrt{2}}$ for all $i\in\range{n}$, and hence $\sigma_{2}\geq\min_{i\in\range{n}}\A_{i}\bar{\x}=\frac{1}{\sqrt{2}}$. On the other hand, for any $\x=(\x_{1},\ldots,\x_{d})\in\B_{d}$, we have
\begin{align}
\min_{i\in\range{n}}\A_{i}\x=\min\Big\{-\frac{1}{\sqrt{2}}\x_{1}+\frac{1}{\sqrt{2}}\x_{l},\frac{1}{\sqrt{2}}\x_{1}+\frac{1}{\sqrt{2}}\x_{l}\Big\}\leq\frac{1}{\sqrt{2}}\x_{l}\leq\frac{1}{\sqrt{2}},
\end{align}
where the first inequality comes from the fact that $\min\{a,b\}\leq\frac{a+b}{2}$ for all $\A,b\in\R$ and the second inequality comes from the fact that $\x\in\B_{d}$ and $|\x_{l}|\leq 1$. As a result, $\sigma_{2}=\max_{\x\in\B_{d}}\min_{i\in\range{n}}\A_{i}\x\leq\frac{1}{\sqrt{2}}$. In conclusion, we have $\sigma_{2}=\frac{1}{\sqrt{2}}$.

\begin{itemize}[leftmargin=*]
\item $\sigma_{1}=\frac{1}{\sqrt{4+2\sqrt{2}}}$.
\end{itemize}

On the one hand, consider $\bar{\x}=\frac{1}{\sqrt{4+2\sqrt{2}}}\vec{e}_{1}+\frac{\sqrt{2}+1}{\sqrt{4+2\sqrt{2}}}\vec{e}_{l}\in\B_{d}$. Then
\begin{align}
\A_{1}\bar{\x}&=-\frac{1}{\sqrt{2}}\cdot\frac{1}{\sqrt{4+2\sqrt{2}}}+\frac{1}{\sqrt{2}}\cdot\frac{\sqrt{2}+1}{\sqrt{4+2\sqrt{2}}}=\frac{1}{\sqrt{4+2\sqrt{2}}}; \\
\A_{i}\bar{\x}&=\frac{1}{\sqrt{2}}\cdot\frac{1}{\sqrt{4+2\sqrt{2}}}+\frac{1}{\sqrt{2}}\cdot\frac{\sqrt{2}+1}{\sqrt{4+2\sqrt{2}}}=\frac{\sqrt{2}+1}{\sqrt{4+2\sqrt{2}}}>\frac{1}{\sqrt{4+2\sqrt{2}}}\quad\forall\,i\in\range{n}/\{1,k\}; \\
\A_{k}\bar{\x}&=1\cdot\frac{1}{\sqrt{4+2\sqrt{2}}}+0\cdot\frac{\sqrt{2}+1}{\sqrt{4+2\sqrt{2}}}=\frac{1}{\sqrt{4+2\sqrt{2}}}.
\end{align}
In all, $\sigma_{1}\geq\min_{i\in\range{n}}\A_{i}\bar{\x}=\frac{1}{\sqrt{4+2\sqrt{2}}}$.

On the other hand, for any $\x=(\x_{1},\ldots,\x_{d})\in\B_{d}$, we have
\begin{align}\label{eqn:perceptron-lower-n-1}
\min_{i\in\range{n}}\A_{i}\x=\min\Big\{-\frac{1}{\sqrt{2}}\x_{1}+\frac{1}{\sqrt{2}}\x_{l},\frac{1}{\sqrt{2}}\x_{1}+\frac{1}{\sqrt{2}}\x_{l},\x_{1}\Big\}.
\end{align}
If $\x_{1}\leq\frac{1}{\sqrt{4+2\sqrt{2}}}$, then \eqn{perceptron-lower-n-1} implies that $\min_{i\in\range{n}}\A_{i}\x\leq\frac{1}{\sqrt{4+2\sqrt{2}}}$; if $\x_{1}\geq\frac{1}{\sqrt{4+2\sqrt{2}}}$, then
\begin{align}
\x_{l}\leq\sqrt{1-\x_{1}^{2}}=\sqrt{1-\frac{1}{4+2\sqrt{2}}}=\frac{\sqrt{2}+1}{\sqrt{4+2\sqrt{2}}},
\end{align}
and hence by \eqn{perceptron-lower-n-1} we have
\begin{align}
\min_{i\in\range{n}}\A_{i}\x\leq-\frac{1}{\sqrt{2}}\x_{1}+\frac{1}{\sqrt{2}}\x_{l}\leq-\frac{1}{\sqrt{2}}\cdot\frac{1}{\sqrt{4+2\sqrt{2}}}+\frac{1}{\sqrt{2}}\cdot\frac{\sqrt{2}+1}{\sqrt{4+2\sqrt{2}}}=\frac{1}{\sqrt{4+2\sqrt{2}}}.
\end{align}
In all, we always have $\min_{i\in\range{n}}\A_{i}\x\leq\frac{1}{\sqrt{4+2\sqrt{2}}}$. As a result, $\sigma_{1}=\max_{\x\in\B_{d}}\min_{i\in\range{n}}\A_{i}\x\leq\frac{1}{\sqrt{4+2\sqrt{2}}}$. In conclusion, we have $\sigma_{1}=\frac{1}{\sqrt{4+2\sqrt{2}}}$.
\\\\\indent
Now, we prove that an $\bar{\x}\in\B_{d}$ satisfying \eqn{perceptron-goal-rewrite} would simultaneously reveal whether $\A$ is from Case 1 or Case 2 as well as the value of $l\in\{2,\ldots,d\}$, by the following algorithm:
\begin{enumerate}
\item Check if one of $\bar{\x}_{2},\ldots,\bar{\x}_{d}$ is larger than 0.94; if there exists an $l'\in\{2,\ldots,d\}$ such that $\bar{\x}_{l'}>0.94$, return `Case 2' and $l=l'$;
\item Otherwise, return `Case 1' and $l=\arg\max_{i\in\{2,\ldots,d\}}\bar{\x}_{i}$.
\end{enumerate}

We first prove that the classification of $\A$ (between Case 1 and Case 2) is correct. On the one hand, assume that $\A$ comes from Case 1. If we wrongly classified $\A$ as from Case 2, we would have $\bar{\x}_{l'}>0.94$ and $\bar{\x}_{1}<\sqrt{1-0.94^{2}}<0.342$; this would imply
\begin{align}
\min_{i\in\range{n}}\A_{i}\bar{\x}=\min\Big\{-\frac{1}{\sqrt{2}}\bar{\x}_{1}+\frac{1}{\sqrt{2}}\bar{\x}_{l},\frac{1}{\sqrt{2}}\bar{\x}_{1}+\frac{1}{\sqrt{2}}\bar{\x}_{l},\bar{\x}_{1}\Big\}\leq \bar{\x}_{1}<\frac{1}{\sqrt{4+2\sqrt{2}}}-0.04\leq\sigma_{1}-\epsilon
\end{align}
by $0.342<\frac{1}{\sqrt{4+2\sqrt{2}}}-0.04$, contradicts with \eqn{perceptron-goal-rewrite}. Therefore, for this case we must make correct classification that $\A$ comes from Case 1.

On the other hand, assume that $\A$ comes from Case 2. If we wrongly classified $\A$ as from Case 1, we would have $\bar{\x}_{l}\leq\max_{i\in\{2,\ldots,d\}}\bar{\x}_{i}\leq 0.94$; this would imply
\begin{align}
\min_{i\in\range{n}}\A_{i}\bar{\x}=\min\Big\{-\frac{1}{\sqrt{2}}\bar{\x}_{1}+\frac{1}{\sqrt{2}}\bar{\x}_{l},\frac{1}{\sqrt{2}}\bar{\x}_{1}+\frac{1}{\sqrt{2}}\bar{\x}_{l}\Big\}\leq\frac{1}{\sqrt{2}}\bar{\x}_{l}<\frac{1}{\sqrt{2}}-0.04\leq\sigma_{2}-\epsilon
\end{align}
by $\frac{0.94}{\sqrt{2}}<\frac{1}{\sqrt{2}}-0.04$, contradicts with \eqn{perceptron-goal-rewrite}. Therefore, for this case we must make correct classification that $\A$ comes from Case 2. In all, our classification is always correct.

It remains to prove that the value of $l$ is correct. If $\A$ is from Case 1, we have
\begin{align}
\sigma_{1}-\epsilon\leq\min_{i\in\range{n}}\A_{i}\bar{\x}=\min\Big\{-\frac{1}{\sqrt{2}}\bar{\x}_{1}+\frac{1}{\sqrt{2}}\bar{\x}_{l},\frac{1}{\sqrt{2}}\bar{\x}_{1}+\frac{1}{\sqrt{2}}\bar{\x}_{l},\bar{\x}_{1}\Big\};
\end{align}
as a result, $\bar{\x}_{1}\geq\sigma_{1}-\epsilon>0.38-0.04=0.34$, and
\begin{align}
-\frac{1}{\sqrt{2}}\bar{\x}_{1}+\frac{1}{\sqrt{2}}\bar{\x}_{l}>0.34\quad\Longrightarrow\quad \bar{\x}_{l}>0.34\sqrt{2}+\bar{\x}_{1}>0.34(\sqrt{2}+1)>0.82.
\end{align}
Because $2\cdot 0.82^{2}>1$, $\bar{\x}_{l}$ must be the largest among $\bar{\x}_{2},\ldots,\bar{\x}_{d}$ (otherwise $l'=\arg\max_{i\in\{2,\ldots,d\}}\bar{\x}_{i}$ and $l\neq l'$ would imply $\|\bar{\x}\|^{2}=\sum_{i\in\range{d}}|\bar{\x}_{i}|^{2}\geq \bar{\x}_{l}^{2}+\bar{\x}_{l'}^{2}\geq 2\bar{\x}_{l}^{2}>1$, contradiction). Therefore, Line 2 of our algorithm correctly returns the value of $l$.

If $\A$ is from Case 2, we have
\begin{align}
\sigma_{2}-\epsilon\leq\min_{i\in\range{n}}\A_{i}\bar{\x}=\min\Big\{-\frac{1}{\sqrt{2}}\bar{\x}_{1}+\frac{1}{\sqrt{2}}\bar{\x}_{l},\frac{1}{\sqrt{2}}\bar{\x}_{1}+\frac{1}{\sqrt{2}}\bar{\x}_{l}\Big\}\leq \frac{1}{\sqrt{2}}\bar{\x}_{l},
\end{align}
and hence $\bar{\x}_{l}\geq\sqrt{2}(\sigma_{2}-\epsilon)\geq\sqrt{2}(\frac{1}{\sqrt{2}}-0.04)>0.94$. Because $2\cdot 0.94^{2}>1$, only one coordinate of $\bar{\x}$ could be at least 0.94 and we must have $l=l'$. Therefore, Line 1 of our algorithm correctly returns the value of $l$.

In all, we have proved that an $\epsilon$-approximate solution $\bar{\x}\in\B_{d}$ for \eqn{perceptron-goal-rewrite} would simultaneously reveal whether $\A$ is from Case 1 or Case 2 as well as the value of $l\in\{2,\ldots,d\}$. On the one hand, notice that distinguishing these two cases requires $\Omega(\sqrt{n-2})=\Omega(\sqrt{n})$ quantum queries to $O_{\A}$ for searching the position of $k$ because of the quantum lower bound for search \cite{bennett1997strengths}; therefore, it gives an $\Omega(\sqrt{n})$ quantum lower bound on queries to $O_{\A}$ for returning an $\bar{\x}$ that satisfies \eqn{perceptron-goal-rewrite}. On the other hand, finding the value of $l$ is also a search problem on the entries of $\A$, which requires $\Omega(\sqrt{d-1})=\Omega(\sqrt{d})$ quantum queries to $O_{\A}$ also due to the quantum lower bound for search \cite{bennett1997strengths}. These observations complete the proof of \thm{perceptron-quantum-lower}.
\end{proof}

Because the kernel-based classifier in \sec{kernel} contains the linear classification in \sec{perceptron} as a special case, \thm{perceptron-quantum-lower} implies an $\Omega(\sqrt{n}+\sqrt{d})$ quantum lower bound on the kernel method.

%=========================================
\subsection{Minimum enclosing ball (MEB)}\label{sec:MEB-lower}
Similarly, the input of the MEB problem is a matrix $\A\in\R^{n\times d}$ such that $\A_{i}\in\B_{d}$ for all $i\in\range{n}$, and we are given the quantum oracle $O_{\A}$ such that $O_{\A}|i\>|j\>|0\>=|i\>|j\>|\A_{ij}\>\ \forall\,i\in\range{n}, j\in\range{d}$. The goal of MEB is to approximately solve
\begin{align}\label{eqn:MEB-rewrite}
\sigma_{\MEB}=\min_{\x\in\R^{d}}\max_{i\in\range{n}}\|\x-\A_{i}\|^{2}.
\end{align}
\thm{MEB-quantum} solves this task with high success probability with cost $\tilde{O}\big(\frac{\sqrt{n}}{\epsilon^{4}}+\frac{\sqrt{d}}{\epsilon^{7}}\big)$. In this subsection, we prove a quantum lower bound that matches this upper bound in $n$ and $d$ for constant $\epsilon$:
\begin{theorem}\label{thm:MEB-quantum-lower}
Assume $0<\epsilon<0.01$. Then to return an $\bar{\x}\in\R_{d}$ satisfying
\begin{align}\label{eqn:MEB-goal-rewrite}
\max_{i\in\range{n}}\|\bar{\x}-\A_{i}\|^{2}\leq\min_{\x\in\R^{d}}\max_{i\in\range{n}}\|\x-\A_{i}\|^{2}+\epsilon
\end{align}
with probability at least $2/3$, we need $\Omega(\sqrt{n}+\sqrt{d})$ quantum queries to $O_{\A}$.
\end{theorem}

By \sec{L2-SVM}, \thm{MEB-quantum-lower} also implies an $\Omega(\sqrt{n}+\sqrt{d})$ quantum lower bound on $\ell^{2}$-margin SVMs.

\begin{proof}
We also assume that $\A$ is from one of the two cases in \thm{perceptron-quantum-lower}; see also \eqn{perceptron-lower-construction-1} and \eqn{perceptron-lower-construction-2}. We denote the maximin value in \eqn{MEB-rewrite} of these cases as $\sigma_{\MEB,1}$ and $\sigma_{\MEB,2}$, respectively. We have:

\begin{itemize}[leftmargin=*]
\item $\sigma_{\MEB,2}=\frac{1}{2}$.
\end{itemize}

On the one hand, consider $\bar{\x}=\frac{1}{\sqrt{2}}\vec{e}_{l}$. Then
\begin{align}
\|\bar{\x}-\A_{1}\|^{2}&=\Big(\x_{1}+\frac{1}{\sqrt{2}}\Big)^{2}+\Big(\x_{l}-\frac{1}{\sqrt{2}}\Big)^{2}+\sum_{i\neq 1,l}\x_{i}^{2}=\Big(\frac{1}{\sqrt{2}}\Big)^{2}=\frac{1}{2}; \\
\|\bar{\x}-\A_{i}\|^{2}&=\Big(\x_{1}-\frac{1}{\sqrt{2}}\Big)^{2}+\Big(\x_{l}-\frac{1}{\sqrt{2}}\Big)^{2}+\sum_{i\neq 1,l}\x_{i}^{2}=\Big(\frac{1}{\sqrt{2}}\Big)^{2}=\frac{1}{2}\quad\forall\,i\in\{2,\ldots,n\}.
\end{align}
Therefore, $\|\bar{\x}-\A_{i}\|^{2}=\frac{1}{2}$ for all $i\in\range{n}$, and hence $\sigma_{\MEB,2}\leq\max_{i\in\range{n}}\|\bar{\x}-\A_{i}\|^{2}=\frac{1}{2}$.

On the other hand, for any $\x=(\x_{1},\ldots,\x_{d})\in\R_{d}$, we have
\begin{align}
&\max_{i\in\range{n}}\|\x-\A_{i}\|^{2} \nonumber \\
=\ &\max\Big\{\Big(\x_{1}-\frac{1}{\sqrt{2}}\Big)^{2}+\Big(\x_{l}-\frac{1}{\sqrt{2}}\Big)^{2}+\sum_{i\neq 1,l}\x_{i}^{2},\Big(\x_{1}+\frac{1}{\sqrt{2}}\Big)^{2}+\Big(\x_{l}-\frac{1}{\sqrt{2}}\Big)^{2}+\sum_{i\neq 1,l}\x_{i}^{2}\Big\} \\
\geq\ &\frac{1}{2}\Big[\Big(\x_{1}-\frac{1}{\sqrt{2}}\Big)^{2}+\Big(\x_{l}-\frac{1}{\sqrt{2}}\Big)^{2}\Big]+\frac{1}{2}\Big[\Big(\x_{1}+\frac{1}{\sqrt{2}}\Big)^{2}+\Big(\x_{l}-\frac{1}{\sqrt{2}}\Big)^{2}\Big]+\sum_{i\neq 1,l}\x_{i}^{2} \label{eqn:MEB-lower-1} \\
=\ &\x_{1}^{2}+\Big(\x_{l}-\frac{1}{\sqrt{2}}\Big)^{2}+\sum_{i\neq 1,l}\x_{i}^{2}+\frac{1}{2} \label{eqn:MEB-lower-1.5} \\
\geq\ &\frac{1}{2},
\end{align}
where \eqn{MEB-lower-1} comes from the fact that $\max\{a,b\}\geq\frac{1}{2}(a+b)$ and $\sum_{i\neq 1,l}\x_{i}^{2}\geq 0$. Therefore, $\sigma_{\MEB,2}\geq\frac{1}{2}$. In all, we must have $\sigma_{\MEB,2}=\frac{1}{2}$.

\begin{itemize}[leftmargin=*]
\item $\sigma_{\MEB,1}=\frac{2+\sqrt{2}}{4}$.
\end{itemize}

On the one hand, consider $\bar{\x}=\big(\frac{1}{2}-\frac{\sqrt{2}}{4}\big)\vec{e}_{1}+\frac{\sqrt{2}}{4}\vec{e}_{l}$. Then
\begin{align}
\|\bar{\x}-\A_{1}\|^{2}&=\Big(\x_{1}+\frac{1}{\sqrt{2}}\Big)^{2}+\Big(\x_{l}-\frac{1}{\sqrt{2}}\Big)^{2}+\sum_{i\neq 1,l}\x_{i}^{2}=\Big(\frac{1}{2}+\frac{\sqrt{2}}{4}\Big)^{2}+\Big(\frac{\sqrt{2}}{4}\Big)^{2}=\frac{2+\sqrt{2}}{4}; \\
\|\bar{\x}-\A_{k}\|^{2}&=(\x_{1}-1)^{2}+\x_{l}^{2}+\sum_{i\neq 1,l}\x_{i}^{2}=\Big(\frac{1}{2}+\frac{\sqrt{2}}{4}\Big)^{2}+\Big(\frac{\sqrt{2}}{4}\Big)^{2}=\frac{2+\sqrt{2}}{4}; \\
\|\bar{\x}-\A_{i}\|^{2}&=\Big(\x_{1}-\frac{1}{\sqrt{2}}\Big)^{2}+\Big(\x_{l}-\frac{1}{\sqrt{2}}\Big)^{2}+\sum_{i\neq 1,l}\x_{i}^{2}=\frac{6-3\sqrt{2}}{4}<\frac{2+\sqrt{2}}{4}\ \ \,\forall\,i\in\range{n}/\{1,k\}.
\end{align}
In all, $\sigma_{\MEB,1}\leq\max_{i\in\range{n}}\|\bar{\x}-\A_{i}\|^{2}=\frac{2+\sqrt{2}}{4}$.

On the other hand, for any $\x=(\x_{1},\ldots,\x_{d})\in\R_{d}$, we have
\begin{align}
\max_{i\in\range{n}}\|\x-\A_{i}\|^{2}\geq&\max\Big\{\Big(\x_{1}+\frac{1}{\sqrt{2}}\Big)^{2}+\Big(\x_{l}-\frac{1}{\sqrt{2}}\Big)^{2}+\sum_{i\neq 1,l}\x_{i}^{2},(\x_{1}-1)^{2}+\x_{l}^{2}+\sum_{i\neq 1,l}\x_{i}^{2}\Big\} \\
\geq\ &\frac{1}{2}\Big[\Big(\x_{1}+\frac{1}{\sqrt{2}}\Big)^{2}+\Big(\x_{l}-\frac{1}{\sqrt{2}}\Big)^{2}\Big]+\frac{1}{2}\Big[(\x_{1}-1)^{2}+\x_{l}^{2}\Big]+\sum_{i\neq 1,l}\x_{i}^{2} \\
=\ &\Big[\x_{1}-\Big(\frac{1}{2}-\frac{\sqrt{2}}{4}\Big)\Big]^{2}+\Big(\x_{l}-\frac{\sqrt{2}}{4}\Big)^{2}+\sum_{i\neq 1,l}\x_{i}^{2}+\frac{2+\sqrt{2}}{4} \label{eqn:MEB-lower-2} \\
\geq\ &\frac{2+\sqrt{2}}{4}.
\end{align}
Therefore, $\sigma_{\MEB,2}\geq\frac{2+\sqrt{2}}{4}$. In all, we must have $\sigma_{\MEB,2}=\frac{2+\sqrt{2}}{4}$.
\\\\\indent
Now, we prove that an $\bar{\x}\in\R_{d}$ satisfying \eqn{MEB-goal-rewrite} would simultaneously reveal whether $\A$ is from Case 1 or Case 2 as well as the value of $l\in\{2,\ldots,d\}$, by the following algorithm:
\begin{enumerate}
\item Check if one of $\bar{\x}_{2},\ldots,\bar{\x}_{d}$ is larger than $\frac{3\sqrt{2}}{8}$; if there exists an $l'\in\{2,\ldots,d\}$ such that $\bar{\x}_{l'}>\frac{3\sqrt{2}}{8}$, return `Case 1' and $l=l'$;
\item Otherwise, return `Case 2' and $l=\arg\max_{i\in\{2,\ldots,d\}}\bar{\x}_{i}$.
\end{enumerate}

We first prove that the classification of $\A$ (between Case 1 and Case 2) is correct. On the one hand, assume that $\A$ comes from Case 1. If we wrongly classified $\A$ as from Case 2, we would have $\bar{\x}_{l}\leq\max_{i\in\{2,\ldots,d\}}\bar{\x}_{i}\leq\frac{3\sqrt{2}}{8}$. By \eqn{MEB-lower-1.5}, this would imply
\begin{align}
\max_{i\in\range{n}}\|\bar{\x}-\A_{i}\|^{2}\geq \Big(\bar{\x}_{l}-\frac{1}{\sqrt{2}}\Big)^{2}+\frac{1}{2}\geq\frac{1}{32}+\frac{1}{2}>\sigma_{\MEB,1}+\epsilon,
\end{align}
contradicts with \eqn{MEB-goal-rewrite}. Therefore, for this case we must make correct classification that $\A$ comes from Case 2.

On the other hand, assume that $\A$ comes from Case 2. If we wrongly classified $\A$ as from Case 1, we would have $\bar{\x}_{l'}>\frac{3\sqrt{2}}{8}$. If $l=l'$, then \eqn{MEB-lower-2} implies that
\begin{align}
\max_{i\in\range{n}}\|\bar{\x}-\A_{i}\|^{2}\geq\Big(\bar{\x}_{l}-\frac{\sqrt{2}}{4}\Big)^{2}+\frac{2+\sqrt{2}}{4}\geq\frac{1}{32}+\frac{2+\sqrt{2}}{4}>\sigma_{\MEB,2}+\epsilon,
\end{align}
contradicts with \eqn{MEB-goal-rewrite}. If $l\neq l'$, then \eqn{MEB-lower-2} implies that
\begin{align}
\max_{i\in\range{n}}\|\bar{\x}-\A_{i}\|^{2}\geq \bar{\x}_{l'}^{2}+\frac{2+\sqrt{2}}{4}\geq\frac{9}{32}+\frac{2+\sqrt{2}}{4}>\sigma_{\MEB,2}+\epsilon,
\end{align}
also contradicts with \eqn{MEB-goal-rewrite}. Therefore, for this case we must make correct classification that $\A$ comes from Case 1.

In all, our classification is always correct. It remains to prove that the value of $l$ is correct. If $\A$ is from Case 1, by \eqn{MEB-lower-1.5} we have
\begin{align}
\frac{1}{2}+0.01\geq\max_{i\in\range{n}}\|\bar{\x}-\A_{i}\|^{2}\geq \bar{\x}_{1}^{2}+\Big(\bar{\x}_{l}-\frac{1}{\sqrt{2}}\Big)^{2}+\sum_{i\neq 1,l}\bar{\x}_{i}^{2}+\frac{1}{2}.
\end{align}
As a result, $\bar{\x}_{i}\leq 0.1<\frac{3\sqrt{2}}{8}$ for all $i\in\range{n}/\{l\}$ and $\bar{\x}_{l}\geq\frac{1}{\sqrt{2}}-0.1>\frac{3\sqrt{2}}{8}$. Therefore, we must have $l=l'$, i.e., Line 1 of our algorithm correctly returns the value of $l$.

If $\A$ is from Case 2, by \eqn{MEB-lower-2} we have
\begin{align}
\frac{2+\sqrt{2}}{4}+0.01&\geq\max_{i\in\range{n}}\|\bar{\x}-\A_{i}\|^{2} \\
&\geq\Big[\bar{\x}_{1}-\Big(\frac{1}{2}-\frac{\sqrt{2}}{4}\Big)\Big]^{2}+\Big(\bar{\x}_{l}-\frac{\sqrt{2}}{4}\Big)^{2}+\sum_{i\neq 1,l}\bar{\x}_{i}^{2}+\frac{2+\sqrt{2}}{4}.
\end{align}
As a result, $\bar{\x}_{i}\leq 0.1<0.25$ for all $i\in\range{n}/\{1,l\}$, $\bar{\x}_{1}\leq\frac{1}{2}-\frac{\sqrt{2}}{4}+0.1<0.25$, and $\bar{\x}_{l}\geq\frac{\sqrt{2}}{4}-0.1>0.25$. Therefore, we must have $\bar{\x}_l=\max_{i\in\{2,\ldots,d\}}\bar{\x}_{i}$, i.e., Line 1 of our algorithm correctly returns the value of $l$.

In all, we have proved that an $\epsilon$-approximate solution $\bar{\x}\in\R_{d}$ for \eqn{MEB-goal-rewrite} would simultaneously reveal whether $\A$ is from Case 1 or Case 2 as well as the value of $l\in\{2,\ldots,d\}$. On the one hand, notice that distinguishing these two cases requires $\Omega(\sqrt{n-2})=\Omega(\sqrt{n})$ quantum queries to $O_{\A}$ for searching the position of $k$ because of the quantum lower bound for search \cite{bennett1997strengths}; therefore, it gives an $\Omega(\sqrt{n})$ quantum lower bound on queries to $O_{\A}$ for returning an $\bar{\x}$ that satisfies \eqn{MEB-goal-rewrite}. On the other hand, finding the value of $l$ is also a search problem on the entries of $\A$, which requires $\Omega(\sqrt{d-1})=\Omega(\sqrt{d})$ quantum queries to $O_{\A}$ also due to the quantum lower bound for search \cite{bennett1997strengths}. These observations complete the proof of \thm{MEB-quantum-lower}.
\end{proof}

%=========================================
\subsection{Zero-sum games}\label{sec:zero-sum-lower}
Recall that for zero-sum games, we are given an $n$-dimensional anti-symmetric matrix $\A$ normalized by $\max_{i,j\in\range{n}}|\A_{i,j}|\leq 1$, and the goal is to return an $\x\in\Delta_{n}$ such that
\begin{align}\label{eqn:zero-sum-game-x-rewrite}
\A\x\leq\epsilon\cdot\vecc{1}_{n}.
\end{align}
Given the quantum oracle $O_{\A}$ in \eqn{oracle-defn}, i.e., $O_{\A}|i\>|j\>|0\>=|i\>|j\>|\A_{ij}\>\ \forall\,i,j\in\range{n}$, \thm{perceptron-quantum} solves this task with high success probability with cost $\tilde{O}\big(\frac{\sqrt{n}}{\epsilon^{4}}\big)$. We prove a matching quantum lower bound in $n$:

\begin{theorem}\label{thm:zero-sum-lower}
Assume $0<\epsilon<1/3$. Then to return an $\x\in\Delta_{n}$ satisfying \eqn{zero-sum-game-x-rewrite} with probability at least $2/3$, we need $\Omega(\sqrt{n})$ quantum queries to $O_{\A}$.
\end{theorem}

\begin{proof}
Assume that there exists an $k\in\range{n}$ such that
\begin{align}\label{eqn:zero-sum-lower-construction}
\A_{ki}=\left\{
    \begin{array}{ll}
      1 & \text{if } i\neq k \\
      0 & \text{if } i=k
    \end{array} \right.\qquad
\A_{ik}=\left\{
    \begin{array}{ll}
      -1 & \text{if } i\neq k \\
      0 & \text{if } i=k
    \end{array} \right.\qquad
\A_{ij}=0\quad\text{if } i,j\neq k.
\end{align}
Denote $\x=(\x_{1},\ldots,\x_{n})^{\dagger}$. Then \eqn{zero-sum-lower-construction} implies that
\begin{align}
\forall\,i\neq k,\ (\A\x)_{i}&=\sum_{j=1}^{n}\A_{ij}\x_{j}=\A_{ik}\x_{k}=-\x_{k}; \label{eqn:zero-sum-lower-construction-1} \\
(\A\x)_{k}&=\sum_{j=1}^{n}\A_{kj}\x_{j}=\sum_{j\neq k}\x_{j}. \label{eqn:zero-sum-lower-construction-2}
\end{align}
In order to have \eqn{zero-sum-game-x-rewrite}, we need to have $-\x_{k}\leq\epsilon$ and $\sum_{j\neq k}\x_{j}\leq\epsilon$ by \eqn{zero-sum-lower-construction-1} and \eqn{zero-sum-lower-construction-2}, respectively. Because $\x_{j}\geq 0$ for all $j\in\range{n}$ and $\sum_{j=1}^{n}\x_{j}=1$, they imply that $1-\epsilon\leq \x_{k}\leq 1$ and $0\leq \x_{j}\leq\epsilon$ for all $j\in\range{n}/\{k\}$. Therefore, if we could return an $\x\in\Delta_{n}$ satisfying \eqn{zero-sum-game-x-rewrite} with probability at least $2/3$, then we become aware of the value of $k$. On the other hand, this is a search problem on the entries of $\A$, which requires $\Omega(\sqrt{n})$ quantum queries to $O_{\A}$ by the quantum lower bound for search \cite{bennett1997strengths}. In all, this implies that the cost of solving the zero-sum game takes at least $\Omega(\sqrt{n})$ quantum queries or gates.
\end{proof}

%%%%%%%%%%%%%%%%%%%%%%%%%%%%%%%%%%%%%%%%%%%%%%%%%%%%%%%%%%%%%%%%%%%%%%%%%%%%%%

\section{Conclusion}\label{sec:conclusion}
We give quantum algorithms for training linear and kernel-based classifiers with complexity $\tilde{O}(\sqrt{n}+\sqrt{d})$, where $n$ and $d$ are the number and dimension of data points, respectively; furthermore, our quantum algorithms are optimal as we prove matching $\Omega(\sqrt{n}+\sqrt{d})$ quantum lower bounds. Our quantum algorithms take standard entry-wise inputs and give classical outputs with succinct representations. Technically, our result is a demonstration of quantum speed-ups for sampling-based classical algorithms using the technique of amplitude amplification and estimation.

Our paper raises a couple of natural open questions for future work. For instance:
\begin{itemize}
\item Can we improve the dependence in $\epsilon$? Recall our quantum algorithms have worst-case complexity $\tilde{O}\big(\frac{\sqrt{n}}{\epsilon^{4}}+\frac{\sqrt{d}}{\epsilon^{8}}\big)$ whereas the classical complexities in \cite{clarkson2012sublinear} are $\tilde{O}\big(\frac{n}{\epsilon^{2}}+\frac{d}{\epsilon^{2}}\big)$; as a result, the quantum algorithms perform better only when we tolerate a significant error. Technically, this is due to the errors coming from amplitude amplification and the fact that state preparation has to be implemented in all iterations. It would be interesting to check if some clever tricks could be applied to circumventing some dependences in $\epsilon$.

\item Can we solve equilibrium point problems other than classification? Recall that our results in \thm{main} are all formulated as maximin problems where the minimum is taken over $\range{n}$ and the maximum is taken over $\B_{d}$ or $\R_{d}$. It would be interesting to study other type of equilibrium point problems in game theory, learning theory, etc.

\item What happens if we work with more sophisticated data structures such as QRAM or its augmented variants? The preprocessing time of these data structures will likely be at least linear. However, it might be still advantageous to do so, for example,  to reduce the amortized complexity when one needs to perform multiple classification tasks on the same data set. 
\end{itemize}

\section*{Acknowledgements}
This work was supported in part by the U.S. Department of Energy, Office of Science, Office of Advanced Scientific Computing Research, Quantum Algorithms Teams program. TL also received support from IBM PhD Fellowship and QISE-NET Triplet Award (NSF DMR-1747426). XW also received support from the National Science Foundation (grants CCF-1755800 and CCF-1816695).

%%%%%%%%%%%%%%%%%%%%%%%%%%%%%%%%%%%%%%%%%%%%%%%%%%%%%%%%%%%%%%%%%%%%%%%%%%%%%%

\providecommand{\bysame}{\leavevmode\hbox to3em{\hrulefill}\thinspace}

%%%%%%%%%%%%%%%%%%%%%%%%%%%%%%%%%%%%%%%%%%%%%%%%%%%%%%%%%%%%%%%%%%%%%%%%%%%%%%

\end{document}